%% file: main_lori2019.tex
\documentclass[runningheads]{llncs}
\usepackage[all]{xy}
\usepackage{pgf}
\usepackage{changepage}
\usepackage{relsize}
\usepackage{tcolorbox}
\usepackage[utf8]{inputenc}
\usepackage{enumerate}
\usepackage{rotating}
\usepackage{appendix}
\usepackage{bussproofs}
\usepackage{framed} 
\usepackage{tikz}
\usepackage{marvosym}
\usepackage{colortbl}

\usetikzlibrary{decorations.pathreplacing,trees,arrows}
\usetikzlibrary{positioning,arrows,calc,decorations.markings}
\tikzset{
modal/.style={>=stealth',shorten >=1pt,shorten <=1pt,auto,node distance=1.5cm,
semithick},
world/.style={circle,draw,minimum size=0.5cm,fill=gray!15},
point/.style={circle,draw,inner sep=0.5mm,fill=black},
reflexive above/.style={->,loop,looseness=7,in=120,out=60},
reflexive below/.style={->,loop,looseness=7,in=240,out=300},
reflexive left/.style={->,loop,looseness=7,in=150,out=210},
reflexive right/.style={->,loop,looseness=7,in=30,out=330}
}
\tikzstyle{level 1}=[level distance=1.5cm, sibling distance=1.5cm]
\tikzstyle{level 2}=[level distance=1.5cm, sibling distance=2cm]
\tikzstyle{bag} = [text width=8em, text centered]

\usepackage{amssymb}
\usepackage{mathtools}  
\usepackage{graphicx}
\usepackage{xfrac}
\usepackage[colorinlistoftodos]{todonotes}
\usepackage{color}
\usepackage{rotating}
\usepackage[hidelinks]{hyperref}
\hypersetup{pdfstartview={FitH},pdfnewwindow=true}

\definecolor{comk}{RGB}{154, 205, 50}
\definecolor{Gray1}{gray}{0.9}
\definecolor{Gray2}{gray}{0.7}

\newcommand{\IRR}{\mathsf{IRR}}
\newcommand{\biggv}{\rotatebox[origin=c]{90}{\Bigg\{ }}
\newcommand{\biggh}{\Bigg\{}

\newtheorem{innercustomlem}{Lemma}
\newenvironment{customlem}[1]
  {\renewcommand\theinnercustomlem{#1}\innercustomlem}
  {\endinnercustomlem}
  
  \newtheorem{innercustomthm}{Theorem}
\newenvironment{customthm}[1]
  {\renewcommand\theinnercustomthm{#1}\innercustomthm}
  {\endinnercustomthm}

\include{newcommandstit}

\EnableBpAbbreviations

\title{A Neutral Temporal Deontic STIT Logic\thanks{This is a pre-print of an article published in Logic, Rationality, and Interaction. The final authenticated version is available online at: \url{https://doi.org/10.1007/978-3-662-60292-8_25}. Work funded by the projects WWTF MA16-028, FWF I2982 and FWF W1255-N23.}
}
\author{Kees van Berkel\textsuperscript{(\Letter)} \and Tim Lyon}

\institute{Institut f\"ur Logic and Computation, Technische Universit\"at Wien, Austria  \\ \email{\{kees,lyon\}@logic.at}}

\authorrunning{K. van Berkel and T. Lyon}

\begin{document}

\maketitle              
\begin{abstract}
In this work we answer a long standing request for temporal embeddings of deontic STIT logics by introducing the multi-agent STIT logic $\tds$. The logic is based upon atemporal utilitarian STIT logic. Yet, the logic presented here will be neutral: instead of committing ourselves to utilitarian theories, we prove the logic $\dtstit$ sound and complete with respect to relational frames not employing any utilitarian function. We demonstrate how these neutral frames can be transformed into utilitarian temporal frames, while preserving validity. Last, we discuss 
 problems that arise from employing binary utility functions in a temporal setting. 

\end{abstract}


\input{introduction}



\section{A Neutral Temporal Deontic STIT Logic}\label{section2a}
\input{section2a}

\section{Soundness and Completeness of $\dtstit$}\label{section2b}

\input{section2b}

\section{Transformations to Utilitarian Models}\label{section3}
\input{section3a}

\input{section3b}


\input{conclusion}

%
%
 \bibliographystyle{splncs04}
%


\appendix

\input{appendix}

\end{document}

%% file: newcommandstit.tex
\newcommand {\Diam}{\rotatebox[origin=c]{45}{$\Box$}}
\newcommand{\Oi}{\otimes_i} 

\newcommand{\Oa}{\otimes_{Ag}}

\newcommand{\ODi}{\ominus_i}

\newcommand{\btr}{\blacktriangledown}

\newcommand{\Ms}{\mathsf{M}}
\newcommand{\Rs}{\mathsf{R}}
\newcommand{\Ws}{\mathsf{W}}
\newcommand{\Vs}{\mathsf{V}}

\newcommand{\Md}{\mathcal{M}^d}
\newcommand{\Wd}{\mathcal{W}}
\newcommand{\Vd}{\mathcal{V}}
\newcommand{\Rd}{\mathcal{R}}

\newcommand{\utstit}{\mathsf{TUS}}
\newcommand{\dtstit}{\mathsf{TDS}}
\newcommand{\dstit}{\mathsf{DS}}

\newcommand{\tds}{\mathsf{TDS}}

\newcommand{\tstit}{\mathsf{Tstit}}
\newcommand{\xstit}{\mathsf{Xstit}}

\newcommand{\ioa}{(\mathsf{IOA})}

\newcommand{\R}{\mathcal{R}}

\newcommand{\agdia}{\langle i \rangle}

\newcommand{\lb}{\langle}
\newcommand{\rb}{\rangle}
\newcommand{\g}{\mathsf{G}}
\newcommand{\h}{\mathsf{H}}
\newcommand{\f}{\mathsf{F}}
\newcommand{\p}{\mathsf{P}}

\newcommand{\lang}{\mathcal{L}_{\dtstit}}
\newcommand{\der}{\vdash_{\dtstit}}

%% file: introduction.tex
\section{Introduction}

With the increasing integration of automated machines in our everyday lives, the development of formal decision-making tools, which take into account moral and legal considerations, is of critical importance~\cite{ArkBriBel05,Ger15,Goo14}. Unfortunately, one of the fundamental hazards of incorporating ethics into decision-making processes, is the apparent incomparability of quantitative and qualitative information---that is, moral
 problems most often 
 resist quantification \cite{NayKoe03}. 

In contrast, 
 utility functions are useful quantitative tools for the formal analysis of decision-making. 
 Initially formulated 
 in \cite{Ben89}, the influential 
 theory of \textit{utilitarianism} has promoted utility calculation as a ground for \textit{ethical deliberation}: 
 in short, those actions 
 generating 
 highest utility, are 
  the morally right actions.
 For this reason, 
 utilitarianism has proven itself to be a fruitful approach 
 in the field of 
 formal deontic reasoning and 
  multi-agent systems (e.g. \cite{BroRam19,Hor01,Mur04}).  

In particular, in the field of STIT logic---
 agency logics developed primarily for the formal analysis of 
 multi-agent 
  choice-making---the utilitarian approach has received increased attention (e.g. \cite{BroRam19,Mur04}). Unfortunately, each available utility function comes with its own (dis)advantages, giving rise to several 
   puzzles (some of them addressed in~\cite{Hor01,HorPac17}). To avoid such problems, we provide an alternative approach:
 instead of 
 settling these 
  philosophical issues, we develop a neutral 
 formalism that can be appropriated to different utilitarian value assignments.

The paper's contributions 
can be summed up as follows: First, we provide a temporal deontic STIT logic called $\dtstit$ (Sec.~\ref{section2a}). With this logic, we answer a long standing request for temporal embeddings of deontic STIT 
 \cite{BelPerXu01,Hor01,Mur04}. Second, although $\tds$ is based upon the atemporal utilitarian STIT logic from 
  \cite{Mur04}, the semantics of $\tds$ will be neutral: 
  instead of committing to utilitarianism, 
 we prove soundness and completeness of $\tds$ with respect to 
  relational frames not employing any utilitarian function (Sec.~\ref{section2b}). 
  This approach also extends the results in \cite{BalHerTro08,HerSch08,Lor13} by showing that $\dtstit$ can be characterized without using 
  the traditional branching-time (BT+AC) structures  
   (cf. \cite{BelPerXu01}). Third, we show how neutral $\tds$ frames can be transformed into 
 utilitarian frames, while preserving validity (Sec.~\ref{section3}). Last, we 
  discuss the philosophical ramifications of employing available utility functions in the extended, temporal setting. In particular, we will argue that binary utility assignments can turn out to be problematic. 

%% file: section2a.tex
In this section, we introduce the language, semantics, and axiomatization of the temporal deontic STIT logic $\dtstit$. In particular, we 
  provide neutral relational frames characterizing the logic, which omit mention of specific utility functions.
The logic will bring together 
 atemporal deontic STIT logic, presented in \cite{Mur04}, 
%
and the temporal STIT logic from 
 \cite{Lor13}.
 

\begin{definition}[The Language $\mathcal{L}_{\dtstit}$]\label{udtstitlanguage} Let $Ag = \{1,2,...,n\}$ be a finite set of agent labels and let $Var =\{p_1,p_2,p_3...\}$ be a countable set of propositional variables. The language $\mathcal{L}_{\dtstit}$ is given by the following BNF grammar: 
{\small $$\phi ::= p \ | \ \lnot \phi \ | \ \phi \land \phi \ | \ \Box \phi \ |  \ [i] \phi \ | \ [Ag] \phi \ | \ \g \phi \ | \ \h \phi \ | \ \Oi\phi$$}
where $i\in Ag$ and $p \in Var$.


\end{definition}
The logical connectives disjunction $\lor$, implication $\rightarrow$, and bi-conditional $\leftrightarrow$ are defined in the usual way. Let $\bot$ be defined as $p\land \lnot p$ and define $\top$ to be $p\lor\lnot p$. The language consists of single agent STIT operators $[i]$, which are choice-operators describing that `agent $i$ sees to it that', and the grand coalition operator $[Ag]$, expressing `the grand coalition of agents sees to it that'.  Furthermore, it contains a settledness operator $\Box$, which holds true of a formula that is settled true at a moment, and thus, holds true regardless of the choices made by any of the agents at that moment.
 The operators $\g$ and $\h$ have, respectively, the usual temporal interpretation `always going to be' and `always has been'. Last, the operator $\Oi$ expresses `agent $i$ ought to see to it that'. 
We define $\Diam, \langle i \rangle, \langle Ag\rangle$ and $\ODi$ as the \emph{duals} of $\Box,[i],[Ag]$ and $\Oi$, respectively (i.e. $\Diam \phi$ iff $\lnot\Box\lnot\phi$, etc.). Furthermore, let $\f\phi$ iff $\lnot \g\lnot \phi$ and $\p\phi$ iff $\lnot\h\lnot \phi$, expressing 
 `$\phi$ holds somewhere in the future' and `$\phi$ holds somewhere in the past', respectively. Finally, deliberative STIT and deliberative ought are obtained accordingly: $[i]^d\phi$ iff $[i]\phi \land \Diamond\lnot \phi$ and $\Oi^d\phi$ iff $\Oi\phi \land \Diam\lnot\phi$. For a discussion of these operators we refer to \cite{Hor01,Lor13}. 

In line with \cite{BalHerTro08,BerLyo19,HerSch08,Lor13}, we provide 
 relational frames for 
 $\tds$ 
  instead of introducing the traditionally employed, BT+AC frames~(cf.~\cite{BelPerXu01}). Explanations of the individual frame properties of Definition \ref{models_dtstit} can be found 
   below.

\begin{definition}[Relational $\dtstit$ Frames and Models]\label{models_dtstit} A 
$\tds$-frame
 is defined as a tuple $F = (W, \R_{\Box}, \{\R_{[i]} \ | \ i \in Ag\}, \R_{[Ag]}, \R_{\g}, \R_{\h}, \{\R_{\Oi} \ | \ i \in Ag\} )$. Let $\R_{[\alpha]}(w) := \{v\in W | (w,v)\in R_{[\alpha]}\}$ for $[\alpha]\in \mathsf{Boxes}$ where
 $\mathsf{Boxes} := \{\Box, \g,\h,[Ag]\} \cup \{[i] \ | \ i\in Ag\}\cup\{\Oi \ | \ i\in Ag\}$. Let $W$ be a non-empty set of worlds $w,v,u...$ and: 
 
\vspace{5pt}
\renewcommand{\arraystretch}{1.1}
\noindent \begin{tabular}{p{1em} p{30pt} p{295pt}}
$\blacktriangleright$ & \multicolumn{2}{l}{For all $i\in Ag$, $\R_{\Box}, \R_{[i]}, \R_{[Ag]} \subseteq W{\times} W$ are equivalence relations such that:}\\
 & {\rm \textbf{(C1)}} & 
 $\R_{[i]} \subseteq \R_{\Box}$. \\
 & {\rm \textbf{(C2)}} & For all $u_{1},...,u_{n} \in W$, if $\R_{\Box}u_{i}u_{j}$ for all $1 \leq i,j \leq n$, then $\bigcap_{i} \R_{[i]}(u_{i}) \neq \emptyset$. \\
 & {\rm \textbf{(C3)}} & For all $w \in W$, $\R_{[Ag]}(w) \subseteq \bigcap_{i \in Ag} \R_{[i]}(w)$. \\
$\blacktriangleright$ & \multicolumn{2}{p{325pt}}{$\R_{\g}\subseteq W\times W$ is a transitive and serial binary relation 
and $\R_{\h}$ is the converse of $\R_{\g}$, such that:}\\ 
 & {\rm \textbf{(T4)}} & For all $w, u, v \in W$, if $\R_{\g}wu$ and $\R_{\g}wv$, then $\R_{\g}uv$, $u = v$, or $\R_{\g}vu$.\\
 & {\rm \textbf{(T5)}} & For all $w, u, v \in W$, if $\R_{\h}wu$ and $\R_{\h}wv$, then $\R_{\h}uv$, $u = v$, or $\R_{\h}vu$.\\
 & {\rm \textbf{(T6)}} & $\R_{\g} \circ \R_{\Box} \subseteq \R_{[Ag]} \circ \R_{\g}$ (relation composition $\circ$ is defined as usual).\\
 & {\rm \textbf{(T7)}} & For all $w,u \in W$, if $u \in \R_{\Box}(w)$, then $u \not\in \R_{\g}(w)$.\\
$\blacktriangleright$ & \multicolumn{2}{l}{For all $i\in Ag$, $\R_{\Oi}\subseteq W{\times} W$ are binary relations such that:} \\
 & {\rm \textbf{(D8)}} & 
  $\R_{\Oi}\subseteq\R_{\Box}$.\\
 & {\rm \textbf{(D9)}} & For all $w\in W$ there exists a $v\in W$ such that $\R_{\Box}wv$ and for all $u\in W$, if $\R_{[i]}vu$ then $\R_{\Oi}wu$.\\
 & {\rm \textbf{(D10)}} & For all $w,v,u,z\in W$, if $\R_{\Box}wv, \R_{\Box}wu$ and $\R_{\Oi}uz$, then $\R_{\Oi}vz$.\\
 & {\rm \textbf{(D11)}} & For all $w,v\in W$, if $\R_{\Oi}wv$ then there exists $u\in W$ s.t. $\R_{\Box}wu$, $\R_{[i]}uv$, and for all $z\in W$, if $\R_{[i]}uz$ then $\R_{\Oi}wz$.\\
\end{tabular}

\vspace{3pt}
\noindent A $\dtstit$-model is a tuple $M = (F,V)$ where $F$ is a $\dtstit$-frame and $V$ is a valuation mapping propositional variables to subsets of $W$, that is, 
  $V{:}\ Var \to \mathcal{P}(W)$. 
\end{definition}

We label the properties of Definition \ref{models_dtstit} referring to choice 
 \textbf{(Ci)}, those relating to temporal aspects 
  \textbf{(Ti)}, and those capturing deontic properties 
   \textbf{(Di)}. Observe that, since $\R_{\Box}$ is an equivalence relation, we obtain equivalence classes $\R_{\Box}(w) = \{v \ | \ (w,v)\in \R_{\Box} \}$. Furthermore, by condition \textbf{(C1)} we know that $\R_{[i]}$ is an equivalence relation partitioning the equivalence classes of $\R_{\Box}$. We call $\R_{\Box}(w)$ a \textit{moment} and for each $v$ in a moment $\R_{\Box}(w)$, we refer to $\R_{[i]}(v)$ as a \textit{choice-cell} for agent $i$ at moment $\R_{\Box}(w)$. In the following, we shall frequently refer to moments and choices in the above sense. Condition \textbf{(C2)} captures the pivotal \textit{independence of agents} principle for STIT logics, ensuring that at every moment, any combination of different agents'
 choices is consistent: i.e., simultaneous choices are independent (see \cite[7C.4]{BelPerXu01}
 ). 
 \textbf{(C3)} ensures that all agents acting together is a necessary condition for the grand coalition of agents acting.\footnote{In future work, we aim to study condition (C3) strengthened to equality, as in \cite{Lor13}. In such a setting,  
 completeness is 
 obtained by 
 proving that each $\tds$-frame can be transformed into a frame (satisfying the same formulae) with 
  strengthened (C3); 
   hence, showing that the logic does not distinguish between the two frame classes. 
}



The conditions on $\R_{\g}$ and $\R_{\h}$ establish that the frames we consider are irreflexive, temporal orderings of \textit{moments}. 
Properties \textbf{(T4)} and \textbf{(T5)} guarantee that \textit{histories}---i.e., 
 maximally 
  ordered paths of worlds passing through moments---are linear.  
Condition \textbf{(T6)} ensures the STIT principle of \textit{no choice between undivided histories}: if two time-lines remain undivided at the next moment, no agent has a choice that realizes one time-line and excludes the other (see \cite[7C.3]{BelPerXu01}). Consequently, this principle also ensures that the ordering of moments is linearly closed with respect to the past and allows for branching with respect to the future: in other words, $\tds$-frames are \textit{treelike}.\footnote{The main reason why the grand coalition operator $[Ag]$ is added to our language, is because it will allow us to axiomatize the \textit{no choice between undivided histories} principle (see A25 of Definition \ref{def:axioms}). For a discussion of $[Ag]$ we refer to \cite{Lor13}.} 
 Last, \textbf{(T7)} ensures the temporal irreflexivity of moments; i.e., the future excludes 
  the present. 
 For an elaborate discussion of the temporal 
 frame conditions 
  we refer to \cite{Lor13}.

 Last, the criteria \textbf{(D8)}-\textbf{(D11)} guarantee an essentially agentive characterization of the obligation operator $\Oi$ (cf. the impartial `ought to be' operator in~\cite{Hor01}). Condition \textbf{(D8)} ensures that ideal worlds are confined to moments: i.e., 
  the ideal worlds accessible at a moment neither lie in the future nor in the past. 
  \textbf{(D9)} ensures that, for each agent there is at every moment a choice available that is an ideal choice (cf. the corresponding `ought implies can' axiom $A14$). 
 Furthermore, \textbf{(D10)} expresses that, for each agent, if a world is ideal from the perspective of a particular world at a moment, that world is ideal from the perspective of any world at that moment
  : i.e., ideal worlds are settled upon moments. 
 Condition \textbf{(D11)} captures the idea that every ideal world extends to a complete ideal choice: i.e., no choice contains both ideal and non-ideal worlds. Last, note 
that conditions \textbf{(C2)} and \textbf{(D9)} together ensure that every combination of distinct agents' ideal choices is consistent, i.e., non-empty.


\begin{definition}[Semantics for $\lang$ 
]\label{Semantics_udstit} Let $M$ 
 be a $\dtstit$-model and let $w\in W$ of $M$. 
 The \emph{satisfaction} of a formula $\phi\in \lang$ in $M$ at $w$ is 
 defined accordingly: 
\begin{small}
\begin{multicols}{2}
\begin{itemize}
\itemsep=1pt
\item[1.] $ M,w {\models}\ p$ iff $w \in V(p)$
\item[2.] $ M,w {\models}\ \lnot \phi$ iff $M,w {\not\models}\ \phi$
\item[3.] $ M,w {\models} \phi \land \psi$ iff $ M,w {\models} \phi$ and $ M,w {\models} \psi$
\item[4.] $ M,w {\models}\ \Box \phi$ iff $\forall u {\in}\ \R_{\Box}(w)$, $ M,u {\models}\ \phi$
\item[5.] $ M,w {\models}\ [i] \phi$ iff $\forall u {\in}\ \R_{[i]}(w)$, $M,u {\models}\ \phi$
\end{itemize}
\columnbreak
\begin{itemize}

\item[6.] $ M,w {\models}\ \Oi\phi$ iff $\forall u{\in}\ \R_{\Oi}(w)$, $M,u {\models}\ \phi$
\item[7.] $ M,w {\models} [Ag] \phi$ iff $\forall u {\in} \R_{[Ag]}(w), M,u {\models} \phi$
\item[8.] $ M,w {\models}\ \g \phi$ iff $\forall u {\in}\ \R_{\g}(w)$, $ M,u {\models}\ \phi$
\item[9.] $ M,w {\models}\ \h \phi$ iff $\forall u {\in}\ \R_{\h}(w)$, $ M,u {\models}\ \phi$
\end{itemize}
\end{multicols}
\end{small}
\noindent 
Global truth, validity, and semantic entailment are defined as usual (see~\cite{BlaRijVen01}).
\end{definition}

The axiomatization of $\dtstit$ is a composition of 
\cite{Mur04}, together with 
\cite{Lor13}. (Note that in the language $\mathcal{L}_{\tds}$ each agent label represents a distinct 
 agent.) 
\newpage
\begin{definition}[Axiomatization of $\tds$]\label{def:axioms} For each $i\in Ag$ we have,
\begin{small}
\begin{multicols}{2}
\begin{itemize}
\item[A0] All propositional tautologies. 
\item[A1] $\Box (\phi\rightarrow \psi)\rightarrow (\Box \phi \rightarrow \Box \psi)$, 
\item[A2] $\Box\phi\rightarrow\phi$
\item[A3] $\Diam \phi\rightarrow \Box\Diam\phi$ 
\item[A4] $[i](\phi\rightarrow \psi)\rightarrow ([i]\phi\rightarrow [i]\psi)$
\item[A5] $[i]\phi\rightarrow \phi$
\item[A6] $\lb i \rb \phi \rightarrow [i]\lb i \rb\phi$ 
\item[A7] $[Ag](\phi\rightarrow \psi)\rightarrow ([Ag]\phi\rightarrow [Ag]\psi)$
\item[A8] $[Ag] \phi \rightarrow \phi$
\item[A9] $\lb Ag \rb \phi \rightarrow [Ag] \lb Ag \rb \phi$
\item[A10] $\bigwedge_{0\leq i \leq n} \Diam [i] \phi_k \rightarrow \Diam \bigwedge_{0\leq i \leq n} [i]\phi_k$
\item[A11] $\bigwedge_{1 \leq i \leq n} [i] \phi_{i} \rightarrow [Ag] \bigwedge_{1 \leq i \leq n} \phi_{i}$
\item[A12] $\Oi (\phi\rightarrow\psi) \rightarrow (\Oi\phi\rightarrow \Oi\psi)$ 
\item[A13] $\Box \phi \rightarrow ([i]\phi\land \Oi\phi)$ 
\item[A14] $\Oi \phi\rightarrow \Diam [i]\phi$
\item[A15] $\Diam \Oi \phi \rightarrow \Box\Oi\phi$ 
\item[A16] $\Box([i]\phi \rightarrow [i]\psi)\rightarrow (\Oi\phi\rightarrow \Oi\psi)$ 
\item[A17] $\g (\phi \rightarrow \psi) \rightarrow (\g \phi \rightarrow \g \psi)$
\item[A18] $\g \phi \rightarrow \g\g \phi$
\item[A19] $\g \phi \rightarrow \f \phi$
\item[A20] $\h (\phi \rightarrow \psi) \rightarrow (\h \phi \rightarrow \h \psi)$
\item[A21] $\phi \rightarrow \g \p \phi$
\item[A22] $\phi \rightarrow \h \f \phi$
\item[A23] $\f \p \phi \rightarrow \p \phi \vee \phi \vee \f \phi$
\item[A24] $\p \f \phi \rightarrow \p \phi \vee \phi \vee \f \phi$
\item[A25] $\f \Diamond \phi \rightarrow \lb Ag \rb \f \phi$
\item[R0] 
  ${\vdash_{\tds}} (\psi{\rightarrow}\phi)$ and ${\vdash_{\tds}}\psi$ implies 
   $\vdash_{\tds} \phi$
\item[R1] 
  ${\vdash_{\tds}} \phi$ implies 
   ${\vdash_{\tds}}[\alpha]\phi$, 
  $[\alpha]{\in}\{\Box,\g,\h\}$ 
\item[R2] 
 $\vdash_{\tds} (\Box \lnot p\land \Box(\g p \land \h p))\rightarrow \phi$ implies $\vdash_{\tds} \phi$, given $p\not\in\phi$
\end{itemize}
\end{multicols}
\end{small}
\noindent
A derivation of $\phi$ in $\dtstit$ from a set $\Gamma$, written $\Gamma \der \phi$, is defined in the usual way (See~\cite[Def. 4.4]{BlaRijVen01}). When $\Gamma {=} \emptyset$, we say $\phi$ is a \emph{theorem}, and write $\der \phi$.

\end{definition}

The axioms, $A1{-}A3$, $A4{-}A6$ and $A7{-}9$ express the S5 behavior of $\Box$, $[i]$ (for each $i{\in} Ag$) and $[Ag]$, respectively. $A10$ is the \textit{independence of agents} axiom. $A11$ captures that `all agents acting together implies the grand coalition of agents acting'. 
 $A13$ is a bridge axiom linking $\Oi$ to $\Box$ and $[i]$ to $\Box$ (cf. (C1) and (D8) of Definition \ref{models_dtstit}). $A14$ corresponds to the `ought implies can' principle (cf. (D9) of Definition \ref{models_dtstit}). $A15$ ensures that, when possible, obligatory choices are settled upon moments (cf. (D10) of Definition \ref{models_dtstit}). $A16$ can be understood as a conditional monotonicity principle for ideal choices (cf. (D11) of Definition \ref{models_dtstit}). 
 Axioms $A12$ and $A13$, together with the necessitation rule $R1$, ensure that $\Oi$ is a normal modal operator.

With respect to the temporal axioms, $A17{-}A19$ capture the KD4 behavior of $\g$, whereas, axioms $A21$ and $A22$ ensure that $\h$ is the converse of $\g$. $A23$ and $A24$ capture \textit{connectedness} of histories through moments and 
$A25$ characterizes \textit{no choice between undivided histories}. Last, 
 $R2$ is a 
  variation of Gabbay's irreflex\-ivity rule (the proofs of Theorem \ref{thm:soundness} and~\ref{lm:completeness} give an indication of 
 the rule's functions). 

%% file: section2b.tex
In this section, we prove that $\dtstit$ is sound and complete relative to the class of $\dtstit$-frames. In the next section, we show how such frames are 
 transformable into frames employing utility assignments. This allows one to model and reason about utilitarian scenarios in a more fine-grained manner, while obtaining completeness of the logic without commitment to particular utility functions.


Unless stated otherwise, 
 all proofs in this section can be found in App.~\ref{appendix}. 

\begin{theorem}\textsc{(soundness of $\dtstit$)}\label{thm:soundness} $\forall \phi\in\mathcal{L}_{\dtstit}$, $\vdash_{\dtstit}\phi$ implies $\models \phi$. 
\end{theorem}



 We prove completeness by constructing maximal consistent sets belonging to a special class and build a canonical $\tds$ model adopting methods from \cite{GabHodRey94,Lor13}.


\begin{definition}
\label{def:MCS} A set of formulae 
 $\Gamma\subseteq\mathcal{L}_{\tds}$ is a 
\emph{maximally consistent set (MCS)} 
 iff 
 (i) $\Gamma \not\vdash_{\dtstit} \bot$, and (ii) for any set $\Gamma' \subseteq \mathcal{L}_{\dtstit}$, if $\Gamma \subset \Gamma'$, then $\Gamma' \der \bot$.

\end{definition}

\begin{definition} \textsc{(canonical model for $\dtstit$)} Let $[\alpha]\in \mathsf{Boxes}$ and let $\langle\alpha\rangle$ be the operator dual to $[\alpha]$. We define the \emph{canonical model} to be the tuple $M^{dt} := (W^{dt}, \R^{dt}_{\Box}, \{\R^{dt}_{[i]} \ | \ i \in Ag\}, \R^{dt}_{[Ag]},$ $\R^{dt}_{\g}, \R^{dt}_{\h}, \{\R^{dt}_{\Oi} \ | \ i\in Ag\}, V^{dt})$ such that:
\begin{itemize}
\item $W^{dt} := \{\Gamma \subset \mathcal{L}_{\dtstit} \ | \ \Gamma \text{ is an MCS}\}$;

\item for all $\Gamma,\Delta\in W^{dt}$, $(\Gamma,\Delta) \in \R^{dt}_{[\alpha]}$ iff for all $\phi\in \lang$, if $[\alpha]\phi \in \Gamma$ , then $\phi \in \Delta$ (for each $[\alpha] \in \mathsf{Boxes}$);

\item $V^{dt}$ is a valuation function s.t. $\forall p\in Atom$, $V^{dt}(p) := \{\Delta \in W^{dt} \ | \ p \in \Delta\}$.
\end{itemize}

\end{definition}


\begin{definition}\textsc{(diamond saturated set~\cite{Lor13})} Let $X$ be a set of MCSs 
 and let $\langle \alpha \rangle$ be dual to $[\alpha] \in \mathsf{Boxes}$. We say that $X$ is a \emph{diamond saturated set} iff for all $\Gamma{\in} X$, for each $\langle\alpha\rangle\phi\ {\in}\ \Gamma$ there exists a $\Delta\ {\in}\ X$ such that $\R_{[\alpha]}\Gamma\Delta$ and~$\phi{\in} \Delta$.
 

\end{definition}


In order to ensure that our canonical model will be irreflexive, we introduce a mechanism that allows us to encode 
MCSs 
 with information that impedes 
 reflexive points in the model. 
 We call these encoded 
  sets IRR-theories and restrict our canonical model 
  to 
  consist of these sets 
   only. Last, we %
   use the notation $M|_{X}$ to indicate a model $M$ whose domain is restricted to the set $X$ (see~\cite[Ch.6]{GabHodRey94}). 

\begin{lemma}\label{lm:non-irr-truth-lemma} Let $X$ be a diamond saturated set with $\Gamma \in X$, $\phi \in \lang$, and let $M^{dt} |_{X}$ be the canonical model restricted to $X$. Then, $M^{dt} |_{X}, \Gamma \models \phi$ iff $\phi \in \Gamma$.

\end{lemma}

\begin{proof} Proven in the usual manner by induction on $\phi$ (see~\cite[Lem.~4.70]{BlaRijVen01}).

\end{proof}

Following \cite{Lor13}, we let IRR-theories  
 be those sets of $\tds$ formulae 
  that (i) are maximally consistent, (ii) contain a label 
  $name(p) := \Box \lnot p \land \Box(\g p \land \h p)$, uniquely
  labeling a \textit{moment} and (iii) for any world that is 
reachable through any `zig-zagging' sequence of diamond operators, that is, 
 every zig-zagging formula $\phi$ of the form, 
$$\langle \alpha_1 \rangle(\phi_1\land \langle \alpha_2 \rangle(\phi_2\land ...\land \langle \alpha_n \rangle \phi_n))...)$$
where $\langle \alpha_i \rangle$ is dual to $[\alpha_{i}] \in \mathsf{Boxes}$ with $1 \leq i \leq n$, there exists a corresponding zig-zagging formula $\phi(q)$ (where $q$ is a propositional variable) of the form, 
$$\langle \alpha_1 \rangle(\phi_1\land \langle \alpha_2 \rangle(\phi_2\land ...\land \langle \alpha_n \rangle (\phi_n \land \Box \lnot q \land \Box(\g q\land \h q)))...)$$
labeling 
 reachable worlds. 
  Let us make the above formally
  precise:  

\begin{definition}\textsc{(irr-theory)~\cite{Lor13}}\label{def:irr-theory} Let $\mathsf{Zig}$ be the set of all zig-zagging formulae in $\lang$ and let \emph{name(p)}$ := \Box \neg p \land \Box (\g p \land \h p)$ where $p$ is a propositional variable. 
 A set of formulae $\Gamma$ is called an \emph{IRR-theory} iff the following hold: 
\begin{itemize}

\item $\Gamma$ is a MCS and $name(p) \in \Gamma$, for some propositional variable $p$;

\item if $\phi \in \Gamma \cap \mathsf{Zig}$, then $\phi(q)\in\Gamma$, for some propositional variable $q$.

\end{itemize}
Henceforth, we refer to $\mathsf{IRR}$ as the 
 set of all IRR-theories in $\lang$.
\end{definition} 

\noindent We now present lemmata relevant to the use of 
IRR-theories in 
 canonical models.

\begin{lemma}\label{lm:consistent-formula-in-irr-theory} Let $\phi \in \lang$ be a consistent formula. Then, there exists an IRR-theory $\Gamma$ such that $\phi \in \Gamma$.
\end{lemma}


\begin{lemma}\textsc{(existence lemma)}\label{lm:existence-lemma} Let $\Gamma$ be an IRR-theory and let $\langle \alpha \rangle$ be dual to $[\alpha] \in \mathsf{Boxes}$. For each 
 $\langle \alpha \rangle \phi \in \Gamma$ there exists an IRR-theory $\Delta$ such that $\R_{[\alpha]}\Gamma\Delta$.
\end{lemma}

Subsequently, it must be shown that the canonical model \textit{restricted} to the set $\IRR$ of IRR-theories (i.e., $M^{dt} |_{\mathsf{IRR}}$) is in fact a $\tds$ model (henceforth, we use $W^{dt}$ and $\IRR$ interchangeably). First, we provide 
 lemmata ensuring that the model satisfies the desired temporal and deontic properties of Definition \ref{models_dtstit}. The first two follow from \cite{Lor13} and the 
 latter four results are proven in App.~\ref{appendix}. 


\begin{lemma}[\cite{Lor13}]\textsc{(property (C2))} Let $\Gamma_{1},...,\Gamma_{n} \in \IRR$ 
 such that $\R^{dt}_{\Box}\Gamma_{i}\Gamma_{j}$ 
  for all $1 \leq i,j \leq n$. Then, there exists a 
 $\Delta\in \IRR$ such that $\R^{dt}_{1}\Gamma_{1}\Delta,...,\R^{dt}_{n}\Gamma_{n}\Delta$.

\end{lemma}

\begin{lemma}[\cite{Lor13}]\textsc{(property (T6))} Let $\Gamma,\Sigma,\Pi\in\IRR$ 
 such that $\R^{dt}_{\g}\Gamma\Sigma$ and $\R^{dt}_{\Box}\Sigma\Pi$. Then, there exists a 
 $\Delta\in\IRR$ such that $\R^{dt}_{[Ag]}\Gamma\Delta$ and $\R^{dt}_{\g}\Delta\Pi$.

\end{lemma}

\begin{lemma}\textsc{(property (D9))}\label{lm:D9-property} Let $\Gamma\in\IRR$. 
 Then, there exists a 
 $\Delta \in \IRR$ such that $\R^{dt}_{\Box}\Gamma\Delta$ and for every 
 $\Sigma \in \IRR$, 
  if $\R^{dt}_{[i]}\Delta\Sigma$, then $\R^{dt}_{\Oi}\Gamma\Sigma$.

\end{lemma}

\begin{lemma}\textsc{(property (D11))}\label{lm:D11-property} Let $\Gamma, \Delta\in \IRR$ 
such that $\R^{dt}_{\Oi}\Gamma\Delta$. Then, there exists a 
$\Sigma \in \IRR$ such that $\R^{dt}_{\Box}\Gamma\Sigma$, $\R^{dt}_{[i]}\Sigma\Delta$, and for all $\Pi \in \IRR$, if $\R^{dt}_{[i]}\Sigma\Pi$, then $\R^{dt}_{\Oi}\Gamma\Pi$.

\end{lemma}


\begin{lemma}\label{lm:caonical-model-is-dtstit-model} The canonical model $M^{dt}|_{\mathsf{IRR}}$ belongs to the class of $\dtstit$ models.
\end{lemma}

\begin{theorem}\textsc{(completeness)}\label{lm:completeness}
If $\phi \in \lang$ is a consistent formula, then $\phi$ is satisfiable on a $\dtstit$-model.
\end{theorem}

%% file: section3a.tex
In this section, we 
investigate 
a truth preserving transformation from $\tds$ models to \textit{utilitarian} STIT models, embedded in a temporal language. In particular, we are concerned with the semantic characterization of the  \textit{dominant ought} \cite[Ch.4]{Hor01}. 
  We start with defining 
  the semantic machinery needed to treat these oughts. In particular, we will introduce a utility function $util$ that maps natural numbers (i.e. utilities) to worlds in our domain. In contrast to \cite{Hor01,Mur04}, we do not restrict the assignment of utilities to complete histories where all worlds on a maximal linear path have identical utility. The reason will be addressed at the end of the section, 
  where we discuss a problem related to 
   utility assignments over histories, 
   arising in 
   temporal extensions of 
    STIT. 

  
The pivotal notion involved in the dominant ought is that of a \textit{state}:
 Agent $i$ 
   cannot influence the choices of all other agents and, for this reason, one can regard the joint interaction of all agents excluding $i$, as a state (of nature) for $i$. 
   To be more precise, we define a \textit{state} $\R^s_{[i]}(v)$ \textit{for i at $v$} accordingly,
   \vspace{-5pt}
$$\R^s_{[i]} (v) = \bigcap\limits_{k\in Ag \setminus \{i\}} \!\!\R_k(v)$$
\vspace{-10pt}

\noindent Consequently, all possible combinations of choices available to the agents $Ag {\setminus} \{i\}$, are the different states available at that moment to agent $i$. 

 Subsequently, we define a 
  \textit{preference order} $\leq$ over choices (and subsets thereof). Let 
   $\R_{[i]}(v),\R_{[i]}(z)\subseteq\R_{\Box}(w)$, then weak preference 
    is defined accordingly, 
 $$\R_{[i]}(v)\leq\R_{[i]}(z) \iff \forall v^{\ast}\in\R_{[i]}(v),\forall z^{\ast}\in\R_{[i]}(z), util(v^{\ast})\leq util(z^{\ast})$$
That is, for an agent a choice is weakly preferred over another, when 
 all values of the possible outcomes of the former are at least as high as 
  those of the latter (where $util(v)$ 
   is the number 
   assigned to $v$, etc
  ). 
 Strict preference is defined as, 
$$\R_{[i]}(v)<\R_{[i]}(z) \iff \R_{[i]}(v)\leq\R_{[i]}(z) \land \R_{[i]}(z)\not\leq\R_{[i]}(v)$$

Next, a \textit{dominance order}
 $\preceq$ over choices $\R_{[i]}(v),\R_{[i]}(z){\subseteq}\R_{\Box}(w)$ is 
defined~as, 
\begin{center}
$\R_{[i]}(v)\preceq \R_{[i]}(z) \iff \forall \R^s_{[i]}(x)\subseteq \R_{\Box}(w), \R_{[i]}(v)\cap\R^s_{[i]}(x) \leq \R_{[i]}(z)\cap\R^s_{[i]}(x)$
\end{center}
We say an agent's choice 
 weakly dominates another, 
  if the values of the outcomes of the former are weakly preferred to those of the latter choice, \textit{given any possible state 
  available to that agent}. 
 For a 
  discussion of dominance orderings see \cite[Ch.~4]{Hor01}.  
 Again, in the usual way 
 we obtain \textit{strict dominance},
$$\R_{[i]}(v)\prec \R_{[i]}(z) \iff \R_{[i]}(v)\preceq \R_{[i]}(z) \land \R_{[i]}(z)\not\preceq \R_{[i]}(v)$$

On the basis of the above, we now formally introduce temporal \textit{utilitarian} STIT frames and models, defined over  \textit{relational} Kripke frames. 

\begin{definition}[Relational $\utstit$ Frames and Models]\label{models_utstit} Let $\R_{[\alpha]}(w) := \{v\in W | (w,v)\in R_{\alpha}\}$ for $[\alpha] \in \{\Box, [Ag], \g, \h \} \cup \{[i] | i \in Ag\}$. A relational \emph{Temporal Utilitarian STIT frame ($\utstit$-frame)} is defined as a tuple $F = ( W, \R_{\Box}, \{\R_{[i]} | i \in Ag\}, \R_{[Ag]}, \R_{\g}, \R_{\h},util )$ where $W$ is a non-empty set of worlds $w,v,u...$ and:

\begin{itemize}

\item For all $i\in Ag$, $\R_{\Box}$, $\R_{[i]}$, $\R_{[Ag]} \subseteq W\times W$ are equivalence relations for which conditions {\rm \textbf{(C1)-(C3)}} of Definition \ref{models_dtstit} hold. 
\item $\R_{\g}\subseteq W\times W$ is a transitive and serial binary relation, whereas $\R_{\h}$ is the converse of $\R_{\g}$, and the conditions {\rm \textbf{(T4)-(T7)}} of Definition \ref{models_dtstit} hold.
\item $util: W \mapsto \mathbb{N}$ is a utility function assigning each world in $W$ to a natural. 
\end{itemize}
A $\utstit$-model is 
a tuple $M = (F,V)$ where $F$ is a $\utstit$-frame and $V$ is a valuation function assigning propositional variables to subsets of $W$: i.e., $V{:}\ Var \mapsto \mathcal{P}(W)$.

\end{definition}

Notice that the above $\utstit$ frames only differ from $\tds$ frames through replacing the relations $\R_{\Oi}$ and corresponding conditions {\rm \textbf{(D8)-(D11)} (for each $i\in Ag$)} with the utility function $util$.  We observe that the assignment of utilities to worlds is agent-independent. Nevertheless, since the choices of an agent depend on which worlds are inside the choice-cells available to the agent, the resulting obligations are in fact agent-dependent.  
 Let us define the new semantics: 
 
 \begin{definition}[Semantics of $\utstit$ models]\label{def:sem_tus} Let $M$ be a $\utstit$-model, 
  $w\in W$ of $M$ and $|\!|\phi|\!|_M = \{ w \ |\ M,w \models  \phi \}$. We define \emph{satisfaction} of a formula $\phi\in\mathcal{L}_{\dtstit}$ as follows: 
 \begin{itemize}
 \item Clause (1)-(10) are the same as those from Definition \ref{Semantics_udstit}, with the exception of clause (7), which we replace by the following clause ($7^{\ast}$):




\end{itemize} 
\vspace{-3pt}
\[\begin{array}{lll} 
\hspace{-10pt}
\quad M,w\models \Oi\phi 	&\textit{iff} 	& \forall \R_{[i]}(v)\subseteq\R_{\Box}(w) \text{ if } \R_{[i]}(v) \not\subseteq |\!|\phi|\!| \text{ then } \exists \R_{[i]}(z)\subseteq\R_{\Box}(w) \text{ s.t. } \\
					&		&\text{(i) } \R_{[i]}(v)\prec \R_{[i]}(z), \text{(ii) } \R_{[i]}(z)\subseteq |\!|\phi|\!| \text{ and }\\
					& 		&\text{(iii) } \forall \R_{[i]}(x)\subseteq \R_{\Box}(w), \R_{[i]}(z)\preceq \R_{[i]}(x) \text{ implies } \R_{[i]}(x)\subseteq |\!|\phi|\!|\\ 
\end{array}
\]  

\end{definition}

Clause $(7^{\ast})$ is interpreted accordingly: Agent $i$ ought to see to it that $\phi$ iff for every choice $\R_{[i]}(v)$ available to $i$ that does not guarantee $\phi$ there (i) exists a strictly dominating choice $\R_{[i]}(z)$ that (ii) does guarantee $\phi$ and (iii) every weakly dominating choice $\R_{[i]}(x)$ over $\R_{[i]}(z)$ also guarantees $\phi$. In other words, all choices not guaranteeing $\phi$ are strictly dominated only by choices guaranteeing $\phi$. (We note that clause ($7^{\ast}$) is obtained through an adaption of the definition provided 
 in \cite{Hor01} 
  to relational frames.)  
 We show that the logic $\dtstit$ is also sound and complete with respect to the 
 class of 
  $\utstit$-frames.

\begin{theorem}\textsc{(soundness)}\label{thm:soundness-tus} $\forall \phi\in\lang$, if $\vdash_{\dtstit}\phi$, then $\mathcal{C}^{u}_f \models \phi$.
\end{theorem}

\begin{proof}  We prove 
 by induction on the given derivation of $\phi$ in $\dtstit$. The argument for axioms A0-A6 and A12 is the same as in Theorem \ref{thm:soundness}. The validity of the axioms A7-A11 can be easily checked by applying 
 semantic clause $(7^{\ast})$ of Definition \ref{models_utstit}.
\end{proof}



We now prove that the class $\mathcal{C}^{u}_f$ of $\utstit$-frames characterizes the same set of formulae as the class $\mathcal{C}^d_f$ of $\dtstit$ frames. We prove both directions separately:




\begin{theorem}{\label{lemma_cu_cd}} $\forall \phi\in\mathcal{L}_{\dtstit}$ we have $\mathcal{C}^u_f\models \phi$ implies $\mathcal{C}^d_f\models\phi$.

\end{theorem}

\begin{proof} We prove 
by contraposition assuming $\mathcal{C}^{d}_f \not\models \phi$. Hence, there is 
 a $\tds$-model, $\Md = (\Wd,\Rd_{\Box}, \{\Rd_i | i\in Ag\}, \Rd_{\h},\Rd_{\g},\Rd_{Ag},\{\Rd_{\Oi}|i\in Ag\},\Vd )$ 
such that $\Md,w \models \neg\phi$ for some $w\in \Wd$. We use $\Md$ to construct a model $\Ms$ in $\mathcal{C}^u_f$, such that:
$$\Ms = ( \Ws, \Rs_{\Box},\{\Rs_i | i\in Ag\},\Rs_{\g},\Rs_{\h},\Rs_{Ag},\mathsf{util},\Vs )$$
We show that $\Ms,w' \models \neg \phi$ for some $w' \in \Ws$. To define $\Ms$ let $\Ws := \Wd$, $\Rs_{\Box} := \Rd_{\Box}$, $\Rs_i := \Rd_i$, $\Rs_{\h} := \Rd_{\h}$, $\Rs_{\g} := \Rd_{\g}$, $\Rs_{Ag} := \Rd_{Ag}$, $\Vs(p) := \Vd(p)$ and let $\mathsf{util}$ be a function assigning each 
 $w\in \Ws$ to a natural number, 
satisfying the 
 following~criteria: 

\begin{itemize}
\item[1.] $\forall i\in Ag, \forall w,v,z\in \Wd$, if $v,z\in\Rd_{\Box}(w)$, $v\in \Rd^s_i(w)\setminus\Rd_{\Oi}(w)$, and $z\in\Rd^s_i(w)\cap \Rd_{\Oi}(w)$, then $\mathsf{util}(v) \leq \mathsf{util}(z)$;

\item[2.] $\forall w,v,z\in \Wd$, if $v\in \Rd_{\Box}(w){\setminus} \Rd_{\Oa}(w)$ and $z\in\Rd_{\Oa}(w)$, then $\mathsf{util}(v){<}\mathsf{util}(z)$;

\item[3.] $\forall w,u,z\in W$, if $v,z\in \Rd^s_i(w)\cap\Rd_{\Oi}(w)$, then $\mathsf{util}(v)=\mathsf{util}(z)$;


\end{itemize}
Let $\R_{\Oa} {:=} \bigcap_{i \in Ag} \R_{\Oi}$, we call 
  $\R_{[i]}(v)\subseteq\R_{\Oi}(w)$ an \textit{optimal choice} for agent~$i$. (It can be easily checked that the function $\mathsf{util}$ can be constructed.) 
  
  We state the following useful lemma (the proof of which is found in App.~\ref{appendix}). 

\begin{lemma}\label{lemma_abbrev} The following holds for any $\dtstit$ frame:\\
 \emph{(1)} $\forall v\in\Rd_{\Box}(w), \Rd_{\Box}(w)=\Rd_{\Box}(v)$; \emph{(2)} $ \forall v\in\Rd_i(w), \Rd_i(w)=\Rd_i(v)$;\\
  \emph{(3)} $\forall v\in\Rd^s_i(w), \Rd^s_i(w)=\Rd^s_i(v)$;
\emph{(4)} $ \forall v\in \R_{\Box}(w)$ we get $\R_{\Oi}(v)=\R_{\Oi}(w)$; \\
\emph{(5)} $\forall \R_{[i]}(z)\subseteq\R_{\Box}(w)$, either $\R_{[i]}(z)\subseteq\R_{\Oi}(w)$ or $\R_{[i]}(z)\cap\R_{\Oi}(w) = \emptyset$. 
\end{lemma}

We observe that conditions \textbf{(C1)}-\textbf{(C3)} and \textbf{(T4)}-\textbf{(T7)} will be satisfied in $\Ms$ since all of the relations of $\Md$, with the exception of $\R_{\Oi}$, are identical to those in $\Ms$. Moreover, 
 $\mathsf{util}$ complies with Definition \ref{models_utstit} and so 
 $\Ms$ is in fact a $\utstit$ model. The desired claim will follow if we additionally show that $\forall \psi\in \mathcal{L}_{\dtstit}$ and $\forall w\in \Wd$:
$$
\Md,w\models \psi \iff \Ms,w\models \psi
$$
We prove the claim by induction on the complexity of $\psi$.

\textit{Base Case.} Let $\psi$ be a propositional variable $p$. By the definition of $\Vs$ in $\Ms$ it follows directly that $\Md,w\models p$ iff $w \in \Vd$ iff $w \in \Vs$ iff $\Ms, w \models p$. 

\textit{Inductive Step.} The cases for the propositional connectives and the modalities $[\alpha] \in \{\Box,\h,\g,[Ag]\}\cup\{[i]|i\in Ag\}$ are straightforward. We consider the non-trivial case when $\psi$ is of the form $\Oi\phi$. Let us first prove the left to right direction.

 ($\Longrightarrow$) Assume $\Md,w\models \Oi\phi$. We show 
 that $\Ms,w\models\Oi\phi$. By the semantics 
 for $\Oi$ 
  (Definition \ref{models_utstit}) it suffices 
   to prove that: $\forall \Rs_i(v)\subseteq\Rs_{\Box}(w)$ if $\Rs_i(v) \not\subseteq |\!|\phi|\!|_{\Ms}$, then $\exists \Rs_i(u)\subseteq \Rs_{\Box}(w)$ such that the following three clauses 
   hold: 
(i) $\Rs_i(v) \prec \Rs_i(u)$; (ii) $\Rs_i(u)\ {\subseteq}\ |\!|\phi|\!|_{\Ms}$; and 
 (iii) 
$\forall \Rs_i(x)\ {\subseteq}\ \Rs_{\Box}(w)$, $\Rs_i(u) \preceq \Rs_i(x)$ implies $\Rs_i(x)\ {\subseteq}\ |\!|\phi|\!|_{\Ms}$.
 
Let $\Rs_i(v)\subseteq\Rs_{\Box}(w)$ be 
arbitrary 
and assume that $\Rs_i(v)\not\subseteq |\!|\phi|\!|_{\Ms}$. We 
 prove that there is a 
 $\Rs_i(u)\subseteq\Rs_{\Box}(w)$ for which conditions (i)-(iii) hold. First, we prove the existence of such a $\Rs_i(u)\subseteq\Rs_{\Box}(w)$: 
  By \textbf{(C1)} and \textbf{(D9)} of Definition \ref{models_dtstit}, we know, 
\begin{equation}\label{3}
\exists u\in \Wd \text{ such that } \Rd_i(u)\subseteq \Rd_{\Box}(w) \text{ and } \Rd_i(u)\subseteq\Rd_{\Oi}(w).
\end{equation}
We also 
 know by \textbf{(D9)} that $\forall j\in Ag{\setminus}\{i\}, \exists u_j\in\Rd_{\Box}(w) \text{ such that } \Rd_j(u_j)\subseteq\Rd_{\otimes_j}(w)$. 
By $\ioa$ we know that $\bigcap_{j\in Ag{\setminus}\{i\}}\Rd_j(u_{j}) \cap \Rd_i(u) \neq \emptyset$, i.e., there exists a $u^{\ast}\in \bigcap_{j\in Ag{\setminus}\{i\}}\Rd_j(u_{j}) \cap \Rd_i(u)$. Consequently, we obtain the following statement, 
\begin{equation}\label{4}
u^{\ast}\in \!\!\! \bigcap_{j\in Ag{\setminus}\{i\}}\!\!\!\Rd_{\otimes_j}(w) \cap \Rd_{\Oi}(w) = \Rd_{\Oa}(w).
\end{equation}
Last, by construction of $\Ms$ we know $\Rd_i(u) = \Rs_i(u)$. We show that (i)-(iii) hold:

\textbf{(i)} We show 
 $\Rs_i(v){\prec}\Rs_i(u)$, that is, 
 (a) $\Rs_i(v)\preceq\Rs_i(u)$ and (b) $\Rs_i(u)\not\prec\Rs_i(v)$:  

\textbf{(a)} Recall, $\Rs_i(v){\not\subseteq}\ |\!|\phi|\!|_{\Ms}$, we know $\exists v^{\ast}{\in}\ \Rs_i(v)$ s.t. $\Ms,v^{\ast}\not\models \phi$. By 
 definition of $\Ms$, $v^{\ast}{\in}\ \Rd_i(v)$ and by (IH) we get $\Md,v^{\ast}\not\models \phi$. Consequently, by the assumption that $\Md,w\models\Oi\phi$, and the fact that $\Md,v^{\ast}{\not\models}\  \phi$,
it follows that $v^{\ast}{\not\in}\ \Rd_{\Oi}(w)$. Hence, we know that $\Rd_i(v){\not\subseteq}\ \Rd_{\Oi}(w)$, which implies $\Rd_{\Oi}(w)\cap\Rd_i(v)=\emptyset$ by Lemma~\ref{lemma_abbrev}$-(5)$. 
 Therefore, by this fact along with statement (\ref{3}) above, we know~that, 

\begin{itemize}
\item[]
For all $x,u^{\btr},v^{\btr}\in \Wd$, if $v^{\btr}\in\Rd^s_i(x)\cap\Rd_i(v)$ and $u^{\btr}\in\Rd^s_i(x)\cap\Rd_i(u)$, then $v^{\btr}\in \Rd^s_i(x){\setminus}\Rd_{\Oi}(w)$ and $u^{\btr}\in\Rd^s_i(x)\cap\Rd_{\Oi}(w)$.
\end{itemize}
Let $x,u^{\btr},v^{\btr}\in Wd$ be arbitrary 
 and assume that $v^{\btr}\in\Rd^s_i(x)\cap\Rd_i(v)$ and $u^{\btr}\in\Rd^s_i(x)\cap\Rd_i(u)$. By the statement above, it follows that $v^{\btr}\in \Rd^s_i(x){\setminus}\Rd_{\Oi}(w)$ and $u^{\btr}\in\Rd^s_i(x)\cap\Rd_{\Oi}(w)$, which in conjunction with criterion 1 on the function $\mathsf{util}$ implies that $\mathsf{util}(v^{\btr})\leq \mathsf{util}(u^{\btr})$. Therefore, the following holds, 
\begin{itemize}
\item[] For all $x,u^{\btr},v^{\btr}\in \Wd$, if $v^{\btr}\in\Rd^s_i(x)\cap\Rd_i(v)$ and $u^{\btr}\in\Rd^s_i(x)\cap\Rd_i(u^{\btr})$, then $\mathsf{util}(v^{\btr})\leq \mathsf{util}(u)$.
\end{itemize}
It follows that $\forall \Rd^s_i(x)\subseteq\Rd_{\Box}(w)$, $\Rd^s_i(x)\cap\Rd_i(v) \leq \Rd^s_i(x)\cap\Rd_i(u)$. Hence, by the definition of $\preceq$ and the definition of $\Ms$, we obtain $\Rs_i(v)\preceq\Rs_i(u)$. 

\textbf{(b)} We need to show $\Rs_i(u)\not\preceq \Rs_i(v)$. By definition of $\preceq$, it suffices to show that $\exists x, \exists u^{\btr},\exists v^{\btr}{\in}\ \Ws$ s.t. $\Rs_i(x){\subseteq}\ \Rs_{\Box}(w)$, $u^{\btr}{\in}\ \Rs_i(u)\cap\Rs^s_i(x)$, $v^{\btr}{\in}\ \Rs_i(v)\cap\Rs^s_i(x)$ and $\mathsf{util}(v^{\btr}){<} \mathsf{util}(u^{\btr})$. Consider $\bigcap_{j\in Ag{\setminus}i}\Rd_j(u_{j}) \cap \Rd_i(u) \neq \emptyset$ from statement (\ref{4}). Let $\Rs^s_i(x):= \bigcap_{j\in Ag{\setminus}i}\Rs_j(u_j)$. Clearly, $\Rs^s_i(x)\subseteq\Rs_{\Box}(w)$. By $\ioa$ we know that $\Rd^s_i(x)\cap \Rd_i(v)\neq\emptyset$ (where $\Rd^s_i(x) = \bigcap_{j\in Ag{\setminus}i}\Rd_j(u_j)$), and so, $\Rs^s_i(x)\cap \Rs_i(v)\neq\emptyset$ by the definition of $\Ms$. Therefore, $\exists v^{\btr}\in\Rs^s_i(x)\cap\Rs_i(v)$. Since $u^{\ast}\in\bigcap_{j\in Ag\setminus i}\Rd_j(u_j)\cap\Rd_i(u)$ (see paragraph above statement (\ref{4})), we know that $u^{\ast}\in\bigcap_{j\in Ag\setminus i}\Rs_j(u_j)\cap\Rs_i(u)$, implying that $u^{\ast}\in\Rs^s_i(x)\cap\Rs_i(u)$. Since also $\Rd_i(v)\cap\Rd_{\Oa}(w)=\emptyset$, as derived in part (i), we obtain $v^{\btr}\in\Rd_{\Box}(w)\setminus\Rd_{\Oa}(w)$. By criterion 2 of 
 $\mathsf{util}$, 
 and the facts 
  $v^{\btr}\in\Rd_{\Box}(w)\setminus\Rd_{\Oa}(w)$ and  
   $u^{\ast}\in\Rd_{\Oa}(w)$, by statement (\ref{4}), we have that $\mathsf{util}(v^{\btr})<\mathsf{util}(u^{\ast})$. Therefore, $\Rs_i(u)\not\preceq\Rs_i(v)$.

\textbf{(ii)} By assumption 
  $\Rd_{\Oi}(w){\subseteq} |\!|\phi|\!|_{\Md}$ and 
   statement (\ref{3}) we get $\Rd_i(u)\subseteq \Rd_{\Oi}(w)$. By IH we have $|\!|\phi|\!|_{\Md}{=} |\!|\phi|\!|_{\Ms}$ and since $\Rd_i(u){=}\Rs_i(u)$ we know $\Rs_i(u)\subseteq |\!|\phi|\!|_{\Ms}$.

\textbf{(iii)} We prove the case by contraposition and show that $\forall \Rs_i(x)\subseteq\Rs_{\Box}(w)$, if $\Rs_i(x)\not\subseteq |\!|\phi|\!| $, then $\Rs_i(u) \not\preceq\Rs_i(x)$. Let $\Rs_i(x)$ by an arbitrary choice-cell in $\Rs_{\Box}(w)$ and assume that $\Rs_i(x)\not\subseteq |\!|\phi|\!|_{\Ms}$. We aim to prove that $\Rs_i(u) \not\preceq\Rs_i(x)$. By definition of $\preceq$ it suffices to show that $\exists \Rs^s_i(y)\subseteq \Rs_{\Box}(w)$ such that $\exists u^{\btr}\in \Rs_i(u)\cap\Rs^s_i(y)$, $\exists x^{\btr}\in \Rs_i(x)\cap\Rs^s_i(y)$, and $\mathsf{util}(x^{\btr})<\mathsf{util}(u^{\btr})$. 


By the assumption that $\Rs_i(x)\not\subseteq |\!|\phi|\!|_{\Ms}$, we know $\exists x^{\btr}\in\Rs_i(x)$ such that $\Ms,x^{\btr}\not\models \phi$. Clearly,  $x^{\btr}\in\Rd_i(x)$, and by (IH) we know that $\Md,x^{\btr}\not\models \phi$. 
Since $\Md, w \models \Oi \phi$, we obtain $(w,x^{\btr})\not\in\Rd_{\Oi}$, and by Lemma \ref{lemma_abbrev}$-(5)$ we obtain $\Rd_i(x)\not\subseteq \Rd_{\Oi}(w)$. 

By statement (\ref{4}) we had $u^{\ast}\in\Rd_{\Oa}(w)$ and $u^{\ast}\in\Rd_{\Oi}(w)$. Also, we know $u^{\ast}\in\Rd_i(u)$ by paragraph 
preceding statement (\ref{4}). Since, $u^{\ast}\in \bigcap_{j\in Ag{\setminus}\{i\}}\Rd_j(u_j) \cap \Rd_i(u)$, we also have $u^{\ast}\in \bigcap_{j\in Ag{\setminus}\{i\}}\Rd_j(u_j)$. Let $\Rd^s_i(u^{*}) := \bigcap_{j\in Ag{\setminus}\{i\}}\Rd_j(u_j)$. By $\ioa$ we obtain $\Rd_i(x)\cap \Rd^s_i(u^{*}) \neq \emptyset$, implying that there exists some $x^{\btr}\in \Rd_i(x)\cap\Rd^s_i(u^{*})$. It follows from \textbf{(D9)} and the fact 
  $\Rd_i(x)\not\subseteq \Rd_{\Oi}(w)$ that $x^{\btr}\not\in\Rd_{\Oa}(w)$, which 
   with the fact $u^{\ast}\in\Rd_{\Oa}(w)$, implies by definition of $\mathsf{util}$ (criterion 2) 
    that $\mathsf{util}(x^{\btr})<\mathsf{util}(u^{\ast})$. By the definition of $\Ms$, we have $x^{\btr}\in \Rs_i(x)\cap\Rs^s_i(u^{*}), u^{\ast}\in\Rs_i(u)\cap\Rs^s_i(u^{*})$ and $\mathsf{util}(x^{\btr})<\mathsf{util}(u^{\ast})$, which implies the desired claim. 
    
    ($\Longleftarrow$) We now prove the right to left direction: Assume $\Ms, w\models \Oi\phi$. We reason towards a contradiction by assuming $\Md,w\not\models \Oi\phi$. Hence, there exists a world $v\in \Rd_{\Oi}(w)$ such that $\Md, v\not\models \phi$. By \textbf{(D11)} we obtain $\Rd_{[i]}(v)\subseteq\Rd_{\Oi}(w)$ and hence $\Rd_{[i]}(v) \not\subseteq |\!|\phi|\!|_{\Md}$. By (IH) and the definition of $\Ms$, we obtain $\Rs_i(v)\not\subseteq|\!|\phi|\!|_{\Ms}$. This fact, in conjunction with the assumption 
     $\Ms, w\models \Oi\phi$, implies that there exists some $\Rs_i(z)\subseteq\Rs_{\Box}(w)$ such that the following holds:
(i) $\Rs_i(v)\prec \Rs_i(z)$; (ii) $\Rs_i(z)\subseteq |\!|\phi|\!|_{\Ms}$; and (iii) $\forall \Rs_i(x)\subseteq \Rs_{\Box}(w)$,   $\Rs_i(z)\preceq \Rs_i(x)$ implies $\Rs_i(x)\subseteq |\!|\phi|\!|_{\Ms}$.

By Lemma \ref{lemma_abbrev}$-(5)$ and the fact that $\Rs_i(z) =\Rd_i(z)$, we know that either \textbf{(a)} $\Rd_i(z)\subseteq\Rd_{\Oi}(w)$ holds or \textbf{(b)} $\Rd_i(z)\cap\Rd_{\Oi}(w) =\emptyset$ holds. 

Assume \textbf{(a)}. We know 
 $\Rs_i(v)\prec\Rs_i(z)$ and therefore, $\Rs_i(z)\not\preceq\Rs_i(v)$. Hence, $\exists \Rs^s_i(x)\subseteq\Rs_{\Box}(w),\exists z^{\ast}\in\Rs_i(z)\cap\Rs^s_i(x),\exists v^{\ast}\in\Rs_i(v)\cap\Rs^s_i(x)$ such that $\mathsf{util}(v^{\ast})<\mathsf{util}(z^{\ast})$. We also know 
 $\Rd_i(v) \subseteq\Rd_{\Oi}(w)$ and $\Rd_i(z)\subseteq\Rd_{\Oi}(w)$ and thus we obtain $z^{\ast},v^{\ast}\in\Rd_{\Oi}\cap \Rd^s_i(x)$. 
 Consequently, by the definition of $\mathsf{util}$ (criterion 3),
 we get $\mathsf{util}(v^{\ast})=\mathsf{util}(z^{\ast})$. Contradiction. 

Assume \textbf{(b)}. We know 
 $\Rs_i(v)\prec\Rs_i(z)$ and therefore, $\Rs_i(z)\not\preceq\Rs_i(v)$. Hence, $\exists \Rs^s_i(x)\subseteq\Rs_{\Box}(w),\exists z^{\ast}\in\Rs_i(z)\cap\Rs^s_i(x),\exists v^{\ast}\in \Rs_i(v)\cap\Rs^s_i(x)$ such that $\mathsf{util}(z^{\ast})\not\leq \mathsf{util}(v^{\ast})$. Then, by definition of $\mathsf{util}$ (criterion 1), 
 either (I) $z^{\ast}\not\in \Rd^s_i(x){\setminus}\Rd_{\Oi}(w)$ or (II) $v^{\ast}\not\in\Rd^s_i(x)\cap\Rd_{\Oi}(w)$. Suppose (I), since $z^{\ast}\in \Rs^s_i(x)$ we infer $z^{\ast}\in \Rd^s_i(x)$ and thus conclude $z^{\ast}\in\Rd_{\Oi}(w)$. However, by 
  earlier assumption $\Rd_i(z)\cap\Rd_{\Oi}(w) = \emptyset$ we obtain 
  $z^{\ast}\not\in \Rd_{\Oi}(w)$. Contradiction. Suppose (II), then since $v^{\ast}\in\Rd^s_i(x)$ we infer $v^{\ast}\not\in\Rd_{\Oi}(w)$. However, $\Rd_{[i]}(v)\subseteq \Rd_{\Oi}(w)$.  Contradiction. 
\end{proof}

\begin{corollary}\textsc{(completeness)}\label{thm:completeness-tus} $\forall \phi\in\lang$, if $\mathcal{C}^{u}_f \models \phi$, then $\vdash_{\dtstit}\phi$.
\end{corollary}

\begin{proof} Follows from Theorem \ref{lemma_cu_cd} above, together with Theorem\ref{lm:completeness}.

\end{proof}

\begin{theorem}\label{lemma_cd_cu}
$\forall \phi \in \mathcal{L}_{\dtstit}$, we get $\mathcal{C}^d_f\models \phi$ implies $\mathcal{C}^u_f\models \phi$.
\end{theorem}

\begin{proof} Follows from 
 Theorem \ref{lm:completeness} together with 
  Theorem \ref{thm:soundness-tus}. 
\end{proof}




%% file: section3b.tex
\subsubsection{The Problem with Two-Valued Utility Functions.}

A well studied candidate function for assigning utilities to \textit{histories}, is the \textit{two-valued} approach where the range of utilities is $\{0,1\}$ (e.g. \cite{Hor01,Mur04}). As a concluding remark of the present section, we briefly discuss the philosophical ramifications of using binary utility functions in a temporal setting.

Observe that, at a moment where all worlds have a utility of $1$ (or all $0$), every obligation becomes vacuously satisfied by definition---in such a scenario we would have 
 $\Oi\phi$ iff $\Box\phi$---and every choice for each agent will ensure all optimal outcomes (see clause $(7^{\ast})$ of Definition \ref{def:sem_tus}).\footnote{This also holds when all intersections of choices of agents contain both a $1$ and a $0$.} If in such a scenario, following 
  \cite{Hor01,Mur04}, utilities are assigned to complete histories and thus remain constant through time, all obligations will also be vacuously satisfied at every future moment from thereon (namely, as one moves into the future, the set of histories passing through a moment can only decrease or stay the same). That at such moments all obligations are vacuously satisfied means that no obligation can be violated. Unfortunately, this also implies that at such moments \textit{contrary-to-duty} (CTD) reasoning---i.e., reasoning about obligations that come into being when a previous obligation has been violated---becomes impossible because CTD obligations require the possibility to violate one's obligations in the first place (e.g. see \cite{PraSer94}).


In order to reason with CTD obligations in \textit{temporal} utilitarian STIT logics, we need to ensure that obligations can be violated, that is, we must consider deliberative obligations: $\Oi^d\phi := \Oi\phi\land\lnot\Box\phi$. This means that, for an obligation $\Oi^d\phi$ to hold, there exists a choice that does not guarantee $\phi$ and, by definition, the latter choice must be strictly dominated by (only) $\phi$ choices. In the binary setting this means that for all optimal choices, there is at least one outcome with a strictly higher utility (which must be $1$). 
 Unfortunately, this 
 has a drawback since at such moments 
 \textit{at least} one of the following 
 holds: (1) Worlds in the intersection of all agents acting in accordance with their duty 
 \emph{all have value $1$}. (2) Worlds in the intersection of all agents violating their duty 
 \textit{all have value} $0$. 
 
Relative to the aforementioned, Fig.~\ref{fig:CTD-situations} illustrates the (only) three scenarios possible in a two-agents, two-choices setting: 
 Sub-figure (i) implies the impossibility 
 of future CTD reasoning 
 in all cases in which at least one 
 agent satisfies its 
  obligation. Sub-figure (ii) implies that there is no future CTD possible in every case witnessing at least one 
   agent violating its 
    obligation. 
    Last, sub-figure (iii) indicates that future CTD obligations can only occur if one of the agents satisfies her obligation if and only if the other violates his. (With the impossibility of future CTD reasoning 
    we mean that from that moment onward, all obligations will be 
    vacuously satisfied.) 
All three
 cases are undesirable since 
 they do not allow for future recuperation in those situations in which they clearly should. 


The above exhibits that, although $\Oi$ does not depend on any temporal aspect (e.g. \cite{Mur04}), we can identify 
 utility functions that are less suitable 
  for temporal extensions of 
  STIT. Binary functions relative to moments only, 
   do not cause these problems, although they have their own issues~\cite{Hor01}. In the case where the function ranges 
 over the set of reals, 
  it is possible to assign utilities in such a way that there is always CTD reasoning possible. 
 In future work, we aim to specify such 
  utility functions, making particular use of temporal aspects of $\tds$-frames. 

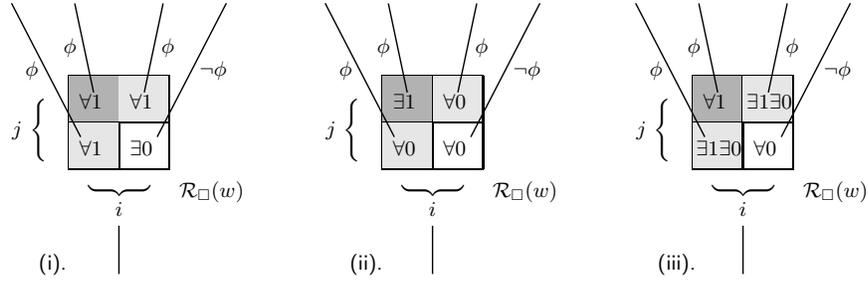
\begin{figure}[t]
\renewcommand{\arraystretch}{1.8}
\newcolumntype{g}{>{\columncolor{Gray}}c}
\hspace{7pt}
\scalebox{0.88}{
\begin{tikzpicture}[modal,font=\sffamily, 
]
\node[bag] (m1) [
] {\begin{tabular}{| p{18.5pt} | p{18.5pt} | }
			\hline 
			 \cellcolor{Gray2}$\ \forall 1  $ & \cellcolor{Gray1}$\ \forall 1 \ $ \\
			\hline
			 \cellcolor{Gray1}$\ \forall 1  $ & $\ \exists 0 \  $ \\
			\hline
		\end{tabular}};
\node[] (root) [below=of m1,yshift=10.5mm
] {};
\node[] (root2) [below=of root,yshift=7mm
] {};
\node[] (m1rb) [above=of root,yshift=-7mm,xshift=5mm] {};
\node[] (m1lb) [above=of root,yshift=-7mm,xshift=-5mm] {};
\node[] (m1lu) [above=of root,yshift=-0.3mm,xshift=-3.5mm] {};
\node[] (m1ru) [above=of root,yshift=-0.3mm,xshift=3.5mm] {};

\node[] (m2) [above=of m1,xshift=-7mm,yshift=-5mm] {};
\node[] (m3) [above=of m1,xshift=7mm,yshift=-5mm] {};
\node[] (m4) [above=of m1,xshift=-17mm,yshift=-5mm] {};
\node[] (m5) [above=of m1,xshift=17mm,yshift=-5mm] {};

\path[-] (root2) edge[color=black] node[above, xshift=14mm,yshift=6mm] {$\R_{\Box}(w)$} 
node[above,xshift=-10mm,yshift=-5mm] {(i).}
node[above,xshift=0mm,yshift=4mm] {$\overset{\biggv}{i}$} 
node[above,xshift=-13.5mm,yshift=12.1mm] {$ j\ \biggh$} (root);

\path[-] (m1lu) edge[color=black] node[left] {$\quad \phi$
} (m2);
\path[-] (m1ru) edge[color=black] node[right] {$ \phi$
} (m3);
\path[-] (m1lb) edge[color=black] node[left] {$\quad \phi$
} (m4);
\path[-] (m1rb) edge[color=black] node[right] {$\lnot\phi$
} (m5);
\end{tikzpicture}
}
\hspace{12pt}
\scalebox{0.88}{
\begin{tikzpicture}[modal,font=\sffamily, 
]
\node[bag] (m1) [
] {\begin{tabular}{| p{18.5pt} | p{18.5pt} | }
			\hline 
			 \cellcolor{Gray2}$\ \exists 1 \ $ & \cellcolor{Gray1}$\ \forall 0 \ $ \\
			\hline
			 \cellcolor{Gray1}$\ \forall 0 \ $ & $\ \forall 0 \  $ \\
			\hline
		\end{tabular}};
\node[] (root) [below=of m1,yshift=10.5mm
] {};
\node[] (root2) [below=of root,yshift=7mm
] {};
\node[] (m1rb) [above=of root,yshift=-7mm,xshift=5mm] {};
\node[] (m1lb) [above=of root,yshift=-7mm,xshift=-5mm] {};
\node[] (m1lu) [above=of root,yshift=-0.3mm,xshift=-3.5mm] {};
\node[] (m1ru) [above=of root,yshift=-0.3mm,xshift=3.5mm] {};

\node[] (m2) [above=of m1,xshift=-7mm,yshift=-5mm] {};
\node[] (m3) [above=of m1,xshift=7mm,yshift=-5mm] {};
\node[] (m4) [above=of m1,xshift=-17mm,yshift=-5mm] {};
\node[] (m5) [above=of m1,xshift=17mm,yshift=-5mm] {};

\path[-] (root2) edge[color=black] node[above, xshift=14mm,yshift=6mm] {$\R_{\Box}(w)$} 
node[above,xshift=-10mm,yshift=-5mm] {(ii).}
 node[above,xshift=0mm,yshift=4mm] {$\overset{\biggv}{i}$} 
node[above,xshift=-13.5mm,yshift=12.1mm] {$ j\ \biggh$} (root);

\path[-] (m1lu) edge[color=black] node[left] {$\quad \phi$
} (m2);
\path[-] (m1ru) edge[color=black] node[right] {$ \phi$
} (m3);
\path[-] (m1lb) edge[color=black] node[left] {$\quad \phi$
} (m4);
\path[-] (m1rb) edge[color=black] node[right] {$\lnot\phi$
} (m5);
\end{tikzpicture}
}
\scalebox{0.88}{
\hspace{12pt}
\begin{tikzpicture}[modal,font=\sffamily, 
] 
\node[bag] (m1) [
] {\begin{tabular}{| p{18.5pt} | p{18.5pt} | }
			\hline 
			 \cellcolor{Gray2}$\ \forall 1 \ $ & \cellcolor{Gray1}$\exists 1\exists 0 \ $ \\
			\hline
			 \cellcolor{Gray1}$ \exists 1\exists 0 $ & $\ \forall 0 \  $ \\
			\hline
		\end{tabular}};
\node[] (root) [below=of m1,yshift=10.5mm
] {};
\node[] (root2) [below=of root,yshift=7mm
] {};
\node[] (m1rb) [above=of root,yshift=-7mm,xshift=5mm] {};
\node[] (m1lb) [above=of root,yshift=-7mm,xshift=-5mm] {};
\node[] (m1lu) [above=of root,yshift=-0.3mm,xshift=-3.5mm] {};
\node[] (m1ru) [above=of root,yshift=-0.3mm,xshift=3.5mm] {};

\node[] (m2) [above=of m1,xshift=-7mm,yshift=-5mm] {};
\node[] (m3) [above=of m1,xshift=7mm,yshift=-5mm] {};
\node[] (m4) [above=of m1,xshift=-17mm,yshift=-5mm] {};
\node[] (m5) [above=of m1,xshift=17mm,yshift=-5mm] {};

\path[-] (root2) edge[color=black] node[above, xshift=14mm,yshift=6mm] {$\R_{\Box}(w)$} 
node[above,xshift=-10mm,yshift=-5mm] {(iii).}
node[above,xshift=0mm,yshift=4mm] {$\overset{\biggv}{i}$} 
node[above,xshift=-13.5mm,yshift=12.1mm] {$ j\ \biggh$} (root);

\path[-] (m1lu) edge[color=black] node[left] {$\quad \phi$
} (m2);
\path[-] (m1ru) edge[color=black] node[right] {$ \phi$
} (m3);
\path[-] (m1lb) edge[color=black] node[left] {$\quad \phi$
} (m4);
\path[-] (m1rb) edge[color=black] node[right] {$\lnot\phi$
} (m5);
\end{tikzpicture}
}
\caption{The only three scenarios where $\otimes_i\phi{\land}\otimes_j\phi{\land}\lnot\Box\phi$ holds true at $\R_{\Box}(w)$ (for $Ag=\{i,j\}$ with 2 choices). Choices of $i$ are vertically presented, those of $j$ horizontally. The symbol $\forall n$ means every history is assigned value $n$, and $\exists n$ means that some history is assigned $n$, for $n \in \{0,1\}$. Optimal choices are shaded and darker shaded when overlapping. At all $\forall k$ outcomes (with $k\in\{0,1\}$), CTD reasoning becomes impossible. }
\label{fig:CTD-situations}
\vspace{-5pt}
\end{figure}


%% file: conclusion.tex
\section{Conclusion and Future Work}\label{conclusion}

In this paper, we extended 
 deontic STIT logic 
   \cite{Mur04} to the temporal setting, incorporating the logic from 
    \cite{Lor13}. In doing so, 
    we answered a long standing open question for temporal embeddings of deontic STIT (e.g. see \cite{BelPerXu01,Hor01,Mur04}). 
    We showed that the resulting logic $\tds$ is sound and complete with respect to its class of 
    frames. 
 We dubbed these frames \emph{neutral} 
  since they allowed us to obtain adequacy of the calculus, while allowing us to refrain from committing 
   to specific 
    utility functions. Subsequently, we showed how these neutral frames can be transformed into particular utilitarian models, while preserving truth. 
  We also briefly argued that 
 in a temporal setting, binary value assignments to histories 
 can generate undesirable behavior 
  with respect to 
   contrary-to-duty obligations.

For future work, we leave open the 
 problem of whether temporal STIT (from \cite{Lor13}) and 
 its deontic extension $\dtstit$ are decidable. 
 Furthermore, we aim to investigate alternative utility assignments that explicitly exploit the temporal aspects of $\tds$; e.g., it might be interesting to consider a dynamic approach taking into account 
  that natural agents have limited foresight relative to (future) utilities. 


%% file: appendix.tex

\section{Proofs}\label{appendix}

\begin{customthm}{\ref{thm:soundness}}\textsc{(soundness)} $\forall \phi\in\mathcal{L}_{\dtstit}$, $\vdash_{\dtstit}\phi$ implies $\models \phi$.
\end{customthm}

\begin{proof} It suffices to show that all axioms are valid and all inference rules preserve validity over the class of $\dtstit$ frames. The rules $R0$ and $R1$, as well as axioms $A0 - A11$, and $A17 - A25$ can be easily checked (See~\cite{Lor13}). We show that $A13 - A16$ are valid and that the $R2$ preserves validity. Let $M$ be an arbitrary $\dtstit$-model with $w$ a world in $M$.

A13. Assume $M,w\models \Box\phi$ and also that $\R_{[i]}wu$ and $\R_{\Oi}wv$. By conditions (C1) and (D8), we know that $\R_{[i]} \subseteq \R_{\Box}$ and $\R_{\Oi} \subseteq \R_{\Box}$, respectively. Therefore, it follows that $\R_{\Box}wu$ and $\R_{\Box}wv$, which implies $M,u \models \phi$ and $M,v \models \phi$ by the assumption. This implies that $M,w \models [i]\phi$ and $M,w \models \Oi \phi$.

A14. Assume $M,w\models \Oi\phi$. By condition (D9), there exists a $v$ such that $\R_{\Box}wv$, and for all $u$ in the model $M$, if $\R_{[i]}vu$, then $\R_{\Oi}wu$. Suppose further that $\R_{[i]}vz$ for an arbitrary $z$; from this, and the previous statement, we may conclude that $\R_{\Oi}wz$ holds, which by the initial assumption implies that $M,z \models \phi$. Therefore, $M,v \models [i] \phi$, and since $\R_{\Box}wv$ holds for some $v$, we have that $M,w \models \Diamond [i] \phi$.

A15. Assume $M,w\models \Diam\Oi\phi$. Thus, there exists a $u$ such that $\R_{\Box}wu$ and $M,u \models \Oi \phi$. Consider an arbitrary $v$ and $z$ such that $\R_{\Box}wv$ and $\R_{\Oi}vz$. By condition (D10), and the fact that $\R_{\Box}wu$, $\R_{\Box}wv$, and $\R_{\Oi}vz$ hold, we may conclude that $\R_{\Oi}uz$ holds. Consequently, $M,z \models \phi$ holds; this fact, in conjunction with the assumption that $\R_{\Box}wv$ and $\R_{\Oi}vz$ hold for arbitrary $v$ and $z$, implies that $M,w \models \Box \Oi \phi$.

A16. Assume $M,w\models \Box([i]\phi\rightarrow [i]\psi)$, $M,w\models \Oi\phi$, and $\R_{\Oi}wu$ for an arbitrary $u$. By condition (D11), the assumption $\R_{\Oi}wu$, implies that there exists a world $v$ such that (i) $\R_{\Box}wv$, (ii) $\R_{[i]}vu$, and (iii) for all $z$, if $\R_{[i]}vz$, then $\R_{\Oi}wz$. The initial assumption, along with fact (i) that $\R_{\Box}wv$, entails that $M,v \models [i] \phi \rightarrow [i] \psi$. Suppose that $\R_{[i]}vx$ for an arbitrary $x$; from fact (iii) we may conclude that $\R_{\Oi}wz$, which with the assumption that $M,w\models \Oi\phi$, implies that $M,z \models \phi$. Hence, $M,v \models [i] \phi$, implying that $M,v \models [i] \psi$. Last, since we know that $\R_{[i]}vu$ by fact (ii), we can conclude that $M,u \models \psi$. Therefore, $M,w\models \Oi\phi \rightarrow \Oi \psi$.

Last, we show \textit{soundness} of the \textbf{IRR}-rule from $\tstit$. Recall the rule:
\begin{center}
\AXC{$\Box \lnot p \land \Box (\g p \land \h p)) \rightarrow \phi$}\RL{if $p$ is atomic and does not occur in $\phi$}
\UIC{$\phi$}
\DP 
\end{center}
We assume that 
 $p$ does not occur in $\phi$. We prove the result by contraposition 
 and assume that $\phi$ is invalid. Therefore, we know there exists a model $M = (F,V)$ s.t. $F$ is a $\dtstit$-frame and $M,w \not\models \phi$ for some $w\in W$ of $M$. We define another $\dtstit$-model $M' = (F,V')$ over the frame $F$ and define the valuation $V'$ as follows:
\[ V'(q) := \begin{cases} 
      V(q) & \text{if } q \neq p, \\
      W \setminus \R_{\Box}(w) & \text{otherwise.}
   \end{cases}
\]
where $\R_{\Box}(w) = \{ v | (w,v)\in \R_{\Box}\}$ (\textit{i.e.} the valuation $V'$ of $p$ contains all worlds except for those sharing the same moment with $w$). Clearly, since $\phi$ does not contain $p$ and the other atomic propositions are valued in the same way in $M$ as in $M'$ we get that $M',w\models \lnot \phi$. However, by the construction of $V'$ and because $F$ is irreflexive by condition (T7), we have that $M',w\models \Box \lnot p \land \Box (\g p\land \h p))$ (the irreflexivity excludes the possibility that for some $u\in \R_{\Box}(w)$, $M',u\models p\land \lnot p$). Since, $M',w \not\models \phi$, by Definition \ref{Semantics_udstit}, we have that $M',w\not\models (\Box \lnot p \land \Box (\g p\land \h p)) \rightarrow \phi$. Hence, we conclude that $(\Box \lnot p \land \Box (\g p\land \h p)) \rightarrow \phi$ is invalid as well. 

\end{proof}

\begin{lemma}\label{lm:properties-of-MCS} Let $\Gamma$ be a MCS. Then, $\Gamma$ has the following properties:
\begin{itemize}

\item $\Gamma \der \phi$ iff $\phi \in \Gamma$;

\item $\phi \in \Gamma$ iff $\neg \phi \not\in \Gamma$;

\item $\phi \land \psi \in \Gamma$ iff $\phi \in \Gamma$ and $\psi \in \Gamma$.

\end{itemize}

\end{lemma}

\begin{proof} We prove each of the claims in turn:
\begin{itemize}

\item [(i)] Assume that $\phi \not\in \Gamma$. Since $\Gamma$ is a maximal, we know that $\Gamma \cup \{\phi\}$ is inconsistent, i.e., $\Gamma \der \phi \rightarrow \bot$. Due to the fact that $\Gamma$ is consistent, we know that $\Gamma \not\der \phi$. For the opposite direction observe that if $\phi \in \Gamma$, then trivially $\Gamma \der \phi$.

\item [(ii)] Suppose that $\phi \in \Gamma$. Observe that if $\neg \phi \in \Gamma$ as well, then $\Gamma$ would be inconsistent; hence, $\neg \phi \not \in \Gamma$. For the backwards direction, assume that $\neg \phi \not\in \Gamma$. If $\phi \not\in \Gamma$ as well, then since $\Gamma$ is a MCS, we know that both $\Gamma \cup \{\phi\} \der \bot$ and $\Gamma \cup \{\neg \phi\} \der \bot$. However, this implies that $\Gamma \der \phi \land \neg \phi$, thus contradicting the consistency of $\Gamma$. This implies that $\phi \in \Gamma$.

\item[(iii)] If $\phi \land \psi \in \Gamma$, then by fact (i) $\phi \in \Gamma$ and $\psi \in \Gamma$ since both $\phi$ and $\psi$ are derivable from $\Gamma$ when $\phi \land \psi \in \Gamma$. The opposite direction is proved similarly.

\end{itemize}

\end{proof}

\begin{lemma}\label{lm:diamond-def-of-canonical-relations} Let $\langle \alpha \rangle$ be dual to $[\alpha] \in \mathsf{Boxes}$. Then, $\R_{[\alpha]}\Gamma\Delta$ iff for all $\phi \in \lang$, if $\phi \in \Delta$, then $\langle \alpha \rangle \phi \in \Gamma$.
\end{lemma}

\begin{proof} Let $\langle \alpha \rangle$ be dual to $[\alpha] \in \mathsf{Boxes}$ and let $\Gamma$ and $\Delta$ be maximally consistent IRR-theories. We prove both directions of the equivalence.

First, assume that $\R_{[\alpha]}\Gamma\Delta$ holds and consider an arbitrary $\phi \in \Delta$. Since $\Delta$ is a MCS, we know that $\neg \phi \not\in \Delta$, which implies by the definition of $\R_{[\alpha]}$ that $[\alpha]\neg\phi \not\in \Gamma$. Due to the fact that $\Gamma$ is a MCS, this implies that$\neg [\alpha] \neg \phi \in \Gamma$, which further implies that $\langle \alpha \rangle \phi \in \Gamma$.

For the opposite direction of the equivalence assume that for all $\phi \in \lang$, if $\phi \in \Delta$, then $\langle \alpha \rangle \phi \in \Gamma$. Let $\psi \in \lang$ and assume that $[\alpha] \psi \in \Gamma$. Then, since $\Gamma$ is a MCS, we know that $\langle \alpha \rangle \neg \psi \not\in \Gamma$. Therefore, $\neg \psi \not\in \Delta$, which implies that $\psi \in \Delta$ since $\Delta$ is a MCS. Since $\psi$ was arbitrary, we have established that $\R_{[\alpha]}\Gamma\Delta$.

\end{proof}








\begin{customlem}{\ref{lm:consistent-formula-in-irr-theory}}
Let $\phi \in \lang$ be a consistent formula. Then, there exists an IRR-theory $\Gamma$ such that $\phi \in \Gamma$.
\end{customlem}

\begin{proof} Let $\phi \in \lang$ be a consistent formula. We enumerate the formulae of $\lang$ so that each formula in odd position is an element of $\mathsf{Zig}$ and make use of this enumeration to build an increasing sequence of consistent theories $\Gamma_{0}$, $\Gamma_{1}$, ..., $\Gamma_{n}$, ...

We let $\Gamma_{0} := \{\phi \land \Box \neg p \land \Box (\g p \land \h p)\}$ for some propositional variable $p$ not occurring in $\phi$. We define the sequence of $\Gamma_{n}$ (for $n > 0$) as follows: Assume that $\Gamma_{n}$ is defined and consider $\psi_{n}$ of the enumeration. We know that either $\Gamma_{n} \cup \{\neg \psi_{n}\}$ is consistent or $\Gamma_{n} \cup \{\psi_{n}\}$ is consistent. If $\Gamma_{n} \cup \{\neg \psi_{n}\}$ is consistent, set $\Gamma_{n+1} := \Gamma_{n} \cup \{\neg \psi_{n}\}$. If $\Gamma_{n} \cup \{\psi_{n}\}$ is consistent, then there are two cases to consider: either (i) $n$ is even or (ii) $n$ is odd. If $n$ is even, then set $\Gamma_{n+1} := \Gamma_{n} \cup \{\psi_{n}\}$. Otherwise, set $\Gamma_{n+1} := \Gamma_{n} \cup \{\psi_{n}, \psi_{n}(q)\}$, where $q$ is a propositional variable not occurring in $\Gamma_{n}$ or $\psi$. We define our desired maximally consistent IRR-theory as follows:
$$\Gamma := \displaystyle{\bigcup_{n \in \mathbb{N}} \Gamma_{n}}$$
To finish the proof we need to show that $\Gamma$ is both a MCS and IRR-theory. We first prove that (i) $\Gamma$ is a MCS and then show that (ii) $\Gamma$ is an IRR-theory.

To prove claim (i), it is useful to first prove that for all $n \in \mathbb{N}$, each $\Gamma_{n}$ is consistent. We show this claim by induction on $n$. In the base case, assume for a contradiction that $\Gamma_{0} = \{\phi \land \Box \neg p \land \Box (\g p \land \h p)\}$ is inconsistent. Hence, $\Box \neg p \land \Box (\g p \land \h p) \land \phi \der \bot$, which further implies that  $\der \Box \neg p \land \Box (\g p \land \h p) \rightarrow (\phi \rightarrow \bot)$. We may infer from the rule R2 that $\der \phi \rightarrow \bot$. However, we know that $\phi$ is consistent, meaning that $\not\der \phi \rightarrow \bot$. We have thus obtained a contradiction implying then that $\Gamma_{0}$ is in fact consistent. For the inductive step assume that $\Gamma_{n}$ is consistent. We want to show that $\Gamma_{n+1}$ is consistent. This trivially follows by the definition of $\Gamma_{n+1}$.

To prove that $\Gamma$ is a MCS, we must show that $\Gamma$ is both consistent and maximal. Assume for a contradiction that $\Gamma$ is inconsistent. Then, this implies that for some finite subset $\Gamma'$ of $\Gamma$, $\Gamma' \vdash \bot$. However, if this is the case, then there exists some $\Gamma_{n}$ such that $\Gamma_{n} \der \bot$. We know that this cannot be the case by the previous paragraph, and so, $\Gamma$ must be consistent. Assume now that there exists some $\Gamma'$ such that $\Gamma \subset \Gamma'$ and $\Gamma' \not\der \bot$. Let $\psi \in \Gamma' \setminus \Gamma$. Since $\psi$ is a formula in $\lang$, we know that if was considered at some point during the construction of the sequence $\Gamma_{0}$, $\Gamma_{1}$, ..., $\Gamma_{n}$, ... Since $\psi \not\in \Gamma$ this implies that there exists some $\Gamma_{m}$ such that $\Gamma_{m} \cup \{\psi\}$ is inconsistent. Therefore, $\Gamma_{m} \der \neg \psi$, which implies that $\Gamma \der \neg \psi$. Due to the fact that $\Gamma \subset \Gamma'$, it follows that $\Gamma' \der \neg \psi$ and $\Gamma' \der \psi$ since $\psi \in \Gamma'$, which is a contradiction. Therefore, $\Gamma$ is a MCS.

We now prove that $\Gamma$ is an IRR-theory. By construction we know that $\phi \land \Box \neg p \land \Box (\g p \land \h p) \in \Gamma_{0} \subset \Gamma$, and since $\Gamma$ is a MCS, it follows that $\Box \neg p \land \Box (\g p \land \h p) \in \Gamma$, thus satisfying the first condition of being an IRR-theory. The second condition of being an IRR-theory is satisfied by the fact that whenever a formula $\psi \in \mathsf{Zig}$ is added to $\Gamma_{m} \subset \Gamma$, for $m \in \mathbb{N}$, the formula $\psi(q)$ is added as well with $q$ fresh.

\end{proof}

\begin{customlem}{\ref{lm:existence-lemma}} 
Let $\Gamma$ be an IRR-theory and let $\langle \alpha \rangle$ be dual to $[\alpha] \in \mathsf{Boxes}$. For each 
 $\langle \alpha \rangle \phi \in \Gamma$ there exists an IRR-theory $\Delta$ such that $\R_{[\alpha]}\Gamma\Delta$.
\end{customlem}

\begin{proof} Similar to~\cite[Lem. 16]{Lor13}.



\end{proof}









\begin{customlem}{\ref{lm:D9-property}} Let $\Gamma$ be an IRR-theory in $W$. Then, there exists an IRR-theory $\Delta \in W$ such that $\R^{dt}_{\Box}\Gamma\Delta$ and for every IRR-theory $\Sigma \in W^{dt}$, if $\R^{dt}_{[i]}\Delta\Sigma$, then $\R^{dt}_{\Oi}\Gamma\Sigma$.

\end{customlem}


\begin{proof} Let $\Gamma$ be an arbitrary IRR-theory in $W^{dt}$. Since $\Gamma$ is an IRR-theory, there is a propositional variable $p$ such that $name(p) \in\Gamma$. Define
$$\Delta_0 := \{[i]\phi | \Oi\phi\in \Gamma\}\cup\{\psi|\Box\psi\in \Gamma\}\cup\{name(p)\}.$$


We will prove by contradiction that $\Delta_0$ is consistent and then extend $\Delta_{0}$ to an IRR-theory.

If $\Delta_0$ is inconsistent, then
$$\vdash_{\dtstit} ([i]\phi_{i}\land ...\land [i]\phi_n \land \psi_1\land ...\land \psi_n \land name(p))\rightarrow \bot$$
where $\psi_{1}, \cdots, \psi_{m} \in \{\psi | \Box \psi \in \Gamma \}$ and $[i] \phi_{1}, \cdots, [i] \phi_{k} \in \{[i] \phi | \Oi \phi \in \Gamma \}$. Let $\hat{\phi} = \phi_1\land...\land\phi_n$ and $\hat{\psi} = \psi_1\land...\land\psi_n$. Since, $\vdash_{\dtstit} [i]\hat{\phi} \leftrightarrow [i] \phi_1\land...\land [i]\phi_n$ we get
$$\vdash_{\dtstit} \hat{\psi} \land name(p)\rightarrow \lnot [i] \hat{\phi}$$
By necessitation for $\Box$ and the $\Box$ K-axiom, we get $\der \Box (\hat{\psi} \land name(p))\rightarrow \Box \lnot [i] \hat{\phi}$, which implies $\der \Box \hat{\psi} \land \Box name(p) \rightarrow \lnot \Diam [i] \hat{\phi}$. Clearly, because $\Box\hat{\psi} \in \Gamma$, $name(p) \in \Gamma$ and $\vdash_{\dtstit} name(p)\rightarrow \Box name(p)$, we have that $\Gamma \der \lnot \Diam [i]\hat{\phi}$. This implies that $\lnot \Diam [i]\hat{\phi} \in \Gamma$ since $\Gamma$ is an IRR-theory.


Also, since $\Oi\phi_1,...,\Oi\phi_n\in\Gamma$ we have $\Oi\phi_1\land...\land\Oi\phi_n\in\Gamma$ since $\Gamma$ is an IRR-theory. By $\vdash_{\dtstit} \Oi \hat{\phi} \leftrightarrow \Oi \phi_1\land...\land\Oi\phi_n$ we conclude $\Oi \hat{\phi}\in\Gamma$ as well. Since $\Oi\hat{\phi}\rightarrow \Diam[i]\hat{\phi}\in \Gamma$ because the formula is an instance of axiom A14, we obtain by modus ponens that $\Diam[i]\hat{\phi}\in \Gamma$. Since $\Gamma$ is an IRR-theory (and hence consistent) we obtain a contradiction, which proves that $\Delta_0$ is consistent. 
 
We now extend $\Delta_0$ to an IRR-theory $\Delta$ by first defining an increasing sequence $\Delta_{0}$, $\Delta_{1}$, ..., $\Delta_{n}$, ... of sets of formulae. Suppose that $\Delta_{n}$ is consistent and defined, and enumerate the formulae of $\lang$ so that each formula in odd position is an element of $\mathsf{Zig}$; we aim to define $\Delta_{n+1}$.

Consider the formula $\psi_{n}$. Either, $\Delta_{n} \cup \{\neg \psi_{n}\}$ is consistent or $\Delta_{n} \cup \{\psi_{n}\}$ is consistent. If the former holds, then set $\Delta_{n+1} := \Delta_{n} \cup \{\neg \psi_{n}\}$. If the latter holds, then there are two subcases to consider: either $n$ is even, in which case, we set $\Delta_{n+1} := \Delta_{n} \cup \{\psi_{n}\}$, or $n$ is odd, in which which case, $\Delta_{n} \cup \{\psi_{n}\}$ is consistent and $\psi_{n} \in \mathsf{Zig}$. We show that in the latter subcase we can find a propositional variable $q$ such that $\Delta_{n} \cup \{\psi_{n}, \psi_{n}(q)\}$ is consistent; we then define $\Delta_{n+1} := \Delta_{n} \cup \{\psi_{n}, \psi_{n}(q)\}$.

Observe that
\begin{equation}
\ODi (name(p) \land \bigwedge_{\chi \in \Delta_{n} \setminus \Delta_{0}} \chi \land \psi_{n}) \in \Gamma
\end{equation}
For otherwise, 
$$\Oi ((name(p) \land \bigwedge_{\chi \in \Delta_{n} \setminus \Delta_{0}} \chi) \rightarrow \neg \psi_{n}) \in \Gamma$$
since $\Gamma$ is an IRR-theory and has the properties specified by Lemma \ref{lm:properties-of-MCS}. By the definition of $\Delta_{0}$ it follows that
$$[i] ((name(p) \land \bigwedge_{\chi \in \Delta_{n} \setminus \Delta_{0}} \chi) \rightarrow \neg \psi_{n}) \in \Delta_{n}$$
Using the fact that $\der [i] \theta \rightarrow \theta$ holds for any formula $\theta$, we infer that
$$\Delta_{n} \der (name(p) \land \bigwedge_{\chi \in \Delta_{n} \setminus \Delta_{0}} \chi) \rightarrow \neg \psi_{n}$$
Since
$$\Delta_{n} \der name(p) \land \bigwedge_{\chi \in \Delta_{n} \setminus \Delta_{0}} \chi$$
we may conclude that $\Delta_{n} \der \neg \psi_{n}$, which contradicts the fact that $\Delta_{n} \cup \{\psi_{n}\}$ is consistent. Therefore, since $\Gamma$ is an IRR-theory and (1) holds, we know that
\begin{equation}
\ODi(name(p) \land \bigwedge_{\chi \in \Delta_{n} \setminus \Delta_{0}} \chi \land \psi_{n}(q)) \in \Gamma
\end{equation}
Using this fact, we may prove that $\Delta_{n+1} := \Delta_{n} \cup \{\psi_{n}, \psi_{n}(q)\}$ is consistent, for suppose otherwise. Then, there exist $\zeta_{1}, \cdots, \zeta_{m} \in \{\zeta | \Box \zeta \in \Gamma \}$ and $[i] \xi_{1}, \cdots, [i] \xi_{k} \in \{[i] \xi | \Oi \xi \in \Gamma \}$ such that
$$\der \zeta_{1} \land \cdots \land\zeta_{m} \rightarrow ([i] \xi_{1} \land \cdots \land [i] \xi_{k} \rightarrow \neg (name(p) \land \bigwedge_{\chi \in \Delta_{n} \setminus \Delta_{0}} \chi \land \psi_{n}(q)))$$
By $\Oi$ necessitation and the $\Oi$ K-axiom, we can derive
$$\der \Oi (\zeta_{1} \land \cdots \land\zeta_{m}) \rightarrow \Oi ([i] \xi_{1} \land \cdots \land [i] \xi_{k} \rightarrow \neg (name(p) \land \bigwedge_{\chi \in \Delta_{n} \setminus \Delta_{0}} \chi \land \psi_{n}(q)))$$
Using axiom A13 we obtain
$$\der \Box (\zeta_{1} \land \cdots \land\zeta_{m}) \rightarrow \Oi ([i] \xi_{1} \land \cdots \land [i] \xi_{k} \rightarrow \neg (name(p) \land \bigwedge_{\chi \in \Delta_{n} \setminus \Delta_{0}} \chi \land \psi_{n}(q)))$$
By our assumption and the fact that $\Gamma$ is an IRR-theory, we know that $\Box (\zeta_{1} \land \cdots \land\zeta_{m}) \in \Gamma$, implying that
$$\Oi ([i] \xi_{1} \land \cdots \land [i] \xi_{k} \rightarrow \neg (name(p) \land \bigwedge_{\chi \in \Delta_{n} \setminus \Delta_{0}} \chi \land \psi_{n}(q))) \in \Gamma$$
We infer the following using modal reasoning
$$\Oi [i] (\xi_{1} \land \cdots \land \xi_{k}) \rightarrow \Oi \neg (name(p) \land \bigwedge_{\chi \in \Delta_{n} \setminus \Delta_{0}} \chi \land \psi_{n}(q))) \in \Gamma$$
One can confirm that $\der \Oi \theta \rightarrow \Oi [i] \theta$ (See~\cite{Mur04}) holds for any formula $\theta$, and therefore
$$\Oi (\xi_{1} \land \cdots \land \xi_{k}) \rightarrow \Oi \neg (name(p) \land \bigwedge_{\chi \in \Delta_{n} \setminus \Delta_{0}} \chi \land \psi_{n}(q))) \in \Gamma$$
Our assumption implies that $\Oi (\xi_{1} \land \cdots \land \xi_{k}) \in \Gamma$, and so
$$\Oi \neg (name(p) \land \bigwedge_{\chi \in \Delta_{n} \setminus \Delta_{0}} \chi \land \psi_{n}(q))) \in \Gamma$$
This contradicts (2) and proves that $\Delta_{n} \cup \{\psi_{n} \psi_{n}(q)\}$ is consistent.

It is easy to infer that $\Delta$ is an IRR-theory by an argument similar to Lemma \ref{lm:consistent-formula-in-irr-theory}.

Clearly, $\R^{dt}_{\Box}\Gamma\Delta$ holds by the definition of $\Delta$. Last, let $\Sigma$ be an arbitrary IRR-theory in $W^{dt}$. Assume that $\R^{dt}_{[i]}\Delta\Sigma$ holds and let $\Oi \xi \in \Gamma$. By definition $[i] \xi \in \Delta$, and so, $\xi \in \Sigma$ by the definition of the relation $\R^{dt}_{[i]}$, which completes the proof.

\end{proof}

\begin{customlem}{\ref{lm:D11-property}} Let $\Gamma$ and $\Delta$ be IRR-theories in $W^{dt}$ such that $\R^{dt}_{\Oi}\Gamma\Delta$. Then, there exists an IRR-theory $\Sigma \in W$ such that $\R^{dt}_{\Box}\Gamma\Sigma$, $\R^{dt}_{[i]}\Sigma\Delta$, and for all $\Pi \in W^{dt}$, if $\R^{dt}_{[i]}\Sigma\Pi$, then $\R^{dt}_{\Oi}\Gamma\Pi$.
\end{customlem}


\begin{proof} To prove this lemma, we proceed differently compared to Lemma \ref{lm:D9-property}, making explicit use of the existence lemma (Lemma \ref{lm:existence-lemma}). Let $\Gamma$ and $\Delta$ be IRR-theories in $W^{dt}$ such that $\R^{dt}_{\Oi}\Gamma\Delta$. Then, there is a $name(p)$ for some $p$ such that $name(p)\in \Delta$. Since $\phi\rightarrow \langle i \rangle \phi\in \Delta$ for any $\phi\in \mathcal{L}_{\tds}$ we know $\langle i\rangle name(p)\in\Delta$. Hence, by Lemma \ref{lm:existence-lemma} we know there exists a $\Sigma\in W^{dt}$ for which $\R^{dt}_{[i]}\Delta\Sigma$. First, we show (i) $\R^{dt}_{[i]}\Sigma\Delta$, then we show (ii) $\R^{dt}_{\Box}\Gamma\Sigma$ and last we show (iii) for any $\Pi\in W^{dt}$ for which $\R^{dt}_{[i]}\Sigma\Pi$, we have $\R^{dt}_{\Oi}\Gamma\Pi$. 

(i) Recall $\R^{dt}_{[i]}\Delta\Sigma$, take an arbitrary $[i]\phi\in \Sigma$, it suffices to show that $\phi\in \Delta$. By Lemma \ref{lm:diamond-def-of-canonical-relations}, we know that $\langle i\rangle [i] \phi\in \Delta$. Since $\vdash_{\tds} \langle i\rangle [i]\theta \rightarrow \theta$ for any $\theta\in\mathcal{L}_{\tds}$ (by axiom A5, A6, and propositional reasoning) we obtain $\phi\in\Delta$; hence $\R^{dt}_{[i]}\Sigma\Delta$.

(ii) Assume an arbitrary $\Box\phi\in \Gamma$. We prove that $\phi\in \Sigma$. We know $\vdash_{\tds} \Box\phi\rightarrow \Oi\phi$ (axiom A13). Hence, since $\Gamma$ is an IRR-theory, we obtain $\Oi\phi\in \Gamma$. Furthermore, $\vdash_{\tds} \Oi\phi\rightarrow\Oi[i]\phi$ (See~\cite{Mur04}), and therefore, $\Oi[i]\phi\in\Gamma$. Since $\R^{dt}_{\Oi}\Gamma\Delta$ we get $[i]\phi\in \Delta$ and thus, by the fact that $\R^{dt}_{[i]}\Delta\Sigma$, we know $\phi\in\Sigma$. We conclude $\R^{dt}_{\Box}\Gamma\Sigma$.

(iii) Take an arbitrary $\Pi\in W^{dt}$. Assume $\R^{dt}_{[i]}\Sigma\Pi$. and $\Oi\phi \in\Gamma$. Since $\Oi\phi\rightarrow \Oi[i]\phi\in \Gamma$, $\Oi[i]\phi\in\Gamma$. Furthermore, since $\vdash_{\tds} \Oi[i]\theta\rightarrow \Oi[i][i]\theta$ for any $\theta\in\mathcal{L}_{\tds}$ (A5, A6, R1, A12), we know $\Oi[i][i]\phi\in\Gamma$, and thus $[i][i]\phi\in\Delta$. Consequently, we get $[i]\phi\in\Sigma$ and last $\phi\in\Pi$, giving us $\R^{dt}_{\Oi}\Gamma\Pi$.
\end{proof}

\begin{customlem}{\ref{lm:caonical-model-is-dtstit-model}} The canonical model $M^{dt}|_{\mathsf{IRR}}$ belongs to the class of $\dtstit$ models.
\end{customlem}

\begin{proof} The argument that $M^{dt}|_{\mathsf{IRR}}$ possesses properties $\mathbf{(C1)}, \mathbf{(C2)}, \mathbf{(C3)^{*}}, \mathbf{(T4)}-\mathbf{(T7)}$ is the same as in~\cite[Lem. 19]{Lor13}. Therefore, we need only confirm that the model satisfies conditions $\mathbf{(D8)}$-$\mathbf{(D11)}$.

The fact that $M^{dt}|_{\mathsf{IRR}}$ satisfies conditions $\mathbf{(D9)}$ and $\mathbf{(D11)}$ follows from Lemma \ref{lm:D9-property} and \ref{lm:D11-property}. We additionally prove that $M^{dt}|_{\mathsf{IRR}}$ satisfies conditions $\mathbf{(D8)}$ and $\mathbf{(D10)}$.

 
\textbf{(D8)} Let $\Gamma$ and $\Delta$ be arbitrary IRR-theories. Assume that $\R^{dt}_{\Oi}\Gamma\Delta$ and assume that $\phi\in \Delta$. Hence, by Lemma \ref{lm:diamond-def-of-canonical-relations}, we know that $\ODi\phi\in\Gamma$. Since $\Box\lnot\phi\rightarrow\Oi\lnot\phi\in\Gamma$, we have $\ODi\phi\rightarrow \Diam\phi\in \Gamma$. Hence, $\Diam\phi\in\Gamma$, which implies that $\R^{dt}_{\Box}\Gamma\Delta$.


\textbf{(D10)} Let $\Gamma,\Delta,\Sigma,\Pi\in W^{dt} \cap \mathsf{IRR}$ and assume that $\R^{dt}_{\Box}\Gamma\Delta$, $\R^{dt}_{\Box}\Gamma\Sigma$, and $\R^{dt}_{\Oi}\Sigma\Pi$. We will show that $\R^{dt}_{\Oi}\Delta\Pi$.

Let $\phi\in \mathcal{L}_{\dtstit}$ and assume $\phi\in \Pi$. Then $\ODi\phi\in\Sigma$ and, hence, $\Diam\ODi\phi\in\Gamma$ by Lemma \ref{lm:diamond-def-of-canonical-relations}. Since
$$\vdash_{\dtstit} (\Diam\Oi\phi\rightarrow\Box\Oi\phi )\rightarrow (\Diam\ODi\phi\rightarrow\Box\ODi\phi)$$
and
$$\Diam\Oi\phi\rightarrow\Box\Oi\phi\in \Gamma$$
we may infer that $\Diam\ODi\phi\rightarrow\Box\ODi\phi\in \Gamma$. Due to the fact that $\Diam\ODi\phi\in\Gamma$, we obtain $\Box\ODi\phi\in\Gamma$, and so, $\ODi\phi\in\Delta$. Therefore, $\R^{dt}_{\Oi}(\Delta,\Pi)$. 


\end{proof}

\begin{customthm}{\ref{lm:completeness}}
If $\phi \in \lang$ is a consistent formula, then $\phi$ is satisfiable on a $\dtstit$-model.
\end{customthm}

\begin{proof} Suppose that $\phi \in \lang$ is consistent. By Lemma \ref{lm:consistent-formula-in-irr-theory}, we can extend $\phi$ to an IRR-theory $\Gamma$ such that $\phi \in \Gamma$. By Lemma \ref{lm:existence-lemma}, we know that the set $\mathsf{IRR}$ is a diamond saturated set, and so, by Lemma \ref{lm:non-irr-truth-lemma}, we know that $M^{dt}|_{\mathsf{IRR}}, \Gamma \models \phi$ iff $\phi \in \Gamma$. Hence, we can conclude that $M^{dt}|_{\mathsf{IRR}}, \Gamma \models \phi$. By Lemma \ref{lm:caonical-model-is-dtstit-model} we know that $M^{dt}|_{\mathsf{IRR}}$ is a $\dtstit$-model; therefore, $\phi$ is satisfiable on a $\dtstit$-model.

\end{proof}